\documentclass[reqno]{amsart}
\usepackage{graphicx}
\usepackage{enumerate}
\usepackage{todonotes}
\usepackage[numbers,sort&compress]{natbib}
\usepackage{hyperref}
\usepackage[foot]{amsaddr}
\usepackage[final,color,notref,notcite]{showkeys}
\usepackage{tikz-cd}
\usepackage{tkz-graph}
\usepackage{dsfont}
\usepackage{xcolor}
\usepackage{slashed}
\usepackage{colortbl}

\usepackage[charter]{mathdesign}
\usepackage[scaled=.96,lf]{XCharter}

\theoremstyle{plain}
\newtheorem*{thm*}{Theorem}
\newtheorem{thm}{Theorem}
\newtheorem{cor}[thm]{Corollary}
\newtheorem{lem}[thm]{Lemma}
\newtheorem{prop}[thm]{Proposition}

\theoremstyle{definition}
\newtheorem{defi}{Definition}
\newtheorem*{exa}{Example}

\theoremstyle{remark}
\newtheorem*{rem}{\bf Remark}

\definecolor{labelkey}{rgb}{0,.56,.7}

\let\bs\boldsymbol
\let\ox\otimes

\let\ol\overline
\newcommand{\sox}{\slashed{\ox}}
\mathchardef\myhyph="2D

\DeclareMathAlphabet{\pazocal}{OMS}{zplm}{m}{n}   

\newcommand{\Ocal}{\pazocal{O}}

\def\bbC{\mathbb{C}}
\def\bbZ{\mathbb{Z}}
\def\bbR{\mathbb{R}}
\def\bbN{\mathbb{N}}
\def\bbI{\mathbb{I}}
\def\bbJ{\mathbb{J}}

\newcommand{\dif}[1]{\mathrm{\,d} #1}             

\newcommand{\nn}{\nonumber}

\def\dg{\dagger}
\def\df{\overset{\mathrm{df}}{=}}
\newcommand{\ket}[1]{\mathop{|#1\rangle}\nolimits}

\newcommand{\haf}{\mathop{{\mathrm{haf}}}\nolimits}

\DeclareMathOperator*{\ls2}{Lommel}

\def\a{\alpha}
\def\b{\beta}
\def\g{\gamma}

\def\s{\sigma}

\def\la{\lambda}

\allowdisplaybreaks

\newcommand{\n}{\bs{n}}
\newcommand{\N}{\bs{N}}

\newcommand{\z}{\bs{z}}

\newcommand{\bd}{\bs{d}}

\def\gbs{\mathrm {GBS}}
\def\dgbs{\mathrm {DGBS}}
\def\mdgbs{\mathrm {mDGBS}}

\DeclareSymbolFont{Eulerscripteusm10}{U}{eus}{m}{n}
\SetSymbolFont{Eulerscripteusm10}{bold}{U}{eus}{b}{n}
\DeclareMathSymbol{\euI}{\mathord}{Eulerscripteusm10}{"4A}
\DeclareMathSymbol{\rH}{\mathord}{Eulerscripteusm10}{"48}
\DeclareMathSymbol{\euK}{\mathord}{Eulerscripteusm10}{"4B}
\DeclareMathSymbol{\euR}{\mathord}{Eulerscripteusm10}{"52}
\DeclareMathSymbol{\rS}{\mathord}{Eulerscripteusm10}{"53}

\begin{document}

\title{A duality at the heart of Gaussian boson sampling}
\author{Kamil Br\'adler, Robert Israel, Maria Schuld and Daiqin Su}
\begin{abstract}
Gaussian boson sampling (GBS) is a near-term quantum computation framework that is believed to be classically intractable, but yet rich of potential applications. In this paper we study the intimate relation between distributions defined over classes of samples from a GBS device with graph matching polynomials. For this purpose, we introduce a new graph polynomial called the displaced GBS polynomial, whose coefficients are the coarse-grained photon-number probabilities of an arbitrary undirected graph $G$ encoded in a GBS device. We report a discovery of a duality between the displaced GBS polynomial of $G$ and the matching polynomial of $G\,\square\,P_2(x)$ -- the Cartesian graph product of $G$ with a single weighted edge also known as the prism over~$G$. Besides the surprising insight gained into Gaussian boson sampling, it opens the door for the new way of classically simulating the Gaussian boson sampling device. Furthermore, it motivates the recent success of a new type of coarse-grained quantum statistics used to construct feature maps in [Schuld et al. 2019].
\end{abstract}

\address{Xanadu, Toronto, Canada}
\email{kamil@xanadu.ai}

\date{\today}

\maketitle

\thispagestyle{empty}

\section{Introduction}\label{sec:intro}

A Gaussian boson sampling (GBS) device was introduced in~\cite{hamilton2016gaussian}, building upon previous ideas~\cite{lund2014boson,rahimi2015can} to generalize boson sampling ~\cite{aaronson2011computational} -- a much-studied proposal to realize a classically intractable quantum computation. As with boson sampling, the road to so-called ``quantum supremacy'' is to understand the results of measurements as samples from a distribution that no classical algorithm can sample from. For this purpose, a GBS device~\cite{hamilton2016gaussian} consists of three main components: an array of $M$ single-mode squeezers whose output is injected into an $M$-mode linear interferometer followed by $M$ photon-number resolving detectors. The measurement is in the Fock basis, that is, the device is  counting the number of photons in each mode. The outputs of the measurement, or ``Gaussian boson samples'', are therefore $M$-tuples $\n = (n_1,...,n_M)$ of nonnegative integers $n_i$. We call these $M$-tuples \textit{click patterns} and we distinguish the click patterns in the collision-free ($n_i\leq1$) and collision ($\n$ arbitrary) regime.

Unlike boson sampling, GBS demonstrates something more than just quantum advantage -- it abounds in applications~\cite{huh2015boson,bradler2018graph,banchi2019molecular,arrazola2018using}. One application, which is closely related to this work, has yielded interesting results in practical applications~\cite{schuld2019quantum}: the output statistics of a GBS device programmed to encode a graph~(see also \cite{bradler2018gaussian}) can be used to construct a feature vector for the graph. The feature vectors, in turn, give rise to a graph similarity measure called the \textit{GBS graph kernel}. Compared to known classical graph kernels, the GBS-based kernel performs well in machine learning tasks based on the similarity measure. Furthermore, the method does not require the entire distribution of the GBS device to be resolved via measurements, but considers ``coarse-grained'' output distributions over classes of measurement outcomes, which significantly reduces the costs of estimating the GBS statistics. This approach was pioneered in~\cite{bradler2018graph}, showing how GBS provides a complete set of graph invariants and therefore is able, at least in principle, to decide the graph isomorphism problem.

Here we present and significantly extend the analysis that lead to successful coarse-graining strategies studied empirically in~\cite{schuld2019quantum}, which summarize measurement outcomes to so-called ``orbits'' and ``meta-orbits''. An orbit is the set of all permutations of a click pattern, while meta-orbits are collections of certain types of orbits. The probability of detecting click patterns belonging to a given orbit is closely related to the coefficients of a structure studied in number theory, algebraic combinatorics and physics for more than fifty years: the graph matching polynomial~\cite{farrell1979introduction,heilmann1970monomers,hosoya1971topological,gutman1977acyclic,godsil1981matchings}. The polynomial coefficients count the so-called $r$-matches of a graph -- the number of ways to choose~$r$ disjoint edges in the graph. But it was found that the coefficients also often say something important about a physical system that the graph represents. Probably the most prominent system is an Ising model, where the matching polynomial is closely related to its partition function~\cite{heilmann1970monomers}. The matching polynomial is -- not surprisingly -- an intractable quantity~\cite{jerrum1989approximating,patel2017deterministic,bezakova2018complexity}, and a lot is known about identities~\cite{godsil1981matchings,cvetkovic1988recent,shi2016graph} and dualities~\cite{godsilAlgComb,lass2004matching} of this extensively investigated~\cite{eep} mathematical object.

In our analysis, we uncover a range of details about the intimate link between matching polynomials and Gaussian boson sampling. In Section \ref{sec:GBS} we first define a cousin of the matching polynomial inspired by a GBS device in the collision-free regime, which we call the \textit{GBS polynomial} of a graph $G$. We derive several identities known from matching polynomials which are also satisfied by the GBS polynomial, and report a new type of identity that has no known counterpart for the matching polynomial. These considerations inspire a new classical simulation method of GBS statistics: The entire GBS polynomial can be computed in one step by calculating the hafnian of the prism over~$G$: the graph $G\,\square\,P_2(x)$, where $P_2(x)$ is a single edge with weight $x\in\bbR$. Extending this strategy to the collision regime, we derive the aforementioned meta-orbit or ``$\Delta$ coarse-grained distribution'', whose properties are studied in this paper as well and were successfully applied in~\cite{schuld2019quantum}.

We go one step further in Section \ref{sec:DGBS} and consider the role of displacement in the light modes. To this end, we generalize the GBS polynomial to the displaced GBS (DGBS) polynomial of $G$  both in the collision and collision-free regime and find that any DGBS polynomial can be written in terms of the matching polynomials of $G$ and all its induced subgraphs. This is surprising and provides a further conceptual simplification of the output statistics  description of a general GBS device. But since the number of induced subgraphs grows exponentially it is helpful only in a limited way for the GBS simulation or the coarse-grained probability evaluation. This, however, changes by uncovering our main result: the DGBS polynomial of $G$ is identical to the matching polynomial of the prism graph $G\,\square\,P_2(x)$, thus generalizing the zero-displacement result for the GBS polynomial. The consequences of this ``duality'' relation are manifold and far-reaching. The calculation of coarse-grained GBS statistics captured by the DGBS can be significantly sped-up by computing a single expression for a graph twice the size -- the matching polynomial of $G\,\square\,P_2(x)$. What is more, once the desired matching polynomial is calculated, one can simply insert any displacement value without a costly recalculation for each instance. Section~\ref{sec:mixed} ultimately generalizes this result, helping us classically simulate the output probability statistics of graphs encoded in a realistic GBS device and thus suffering from decoherence -- most notably due to photon loss.

Lastly, Section~\ref{sec:coarsegrainDistro} explores an important prediction following from the introduction of the DGBS polynomial -- a new type of photon number coarse-grained statistics. The main motivation behind this investigation is to find quantities that are useful for applications. That is, we are looking for ways to post-process GBS samples that are ``quantum feasible'' (that is, feasible in practice for realistic parameters of squeezing and displacement), and at the same time classically intractable. The study of GBS polynomials allows us to motivate why meta-orbits used in~\cite{schuld2019quantum} are a potential candidate.

\section{Notation, Preliminaries and a summary of previous results}\label{sec:preliminaries}

We start by recalling the results of~\cite{hamilton2016gaussian} relevant to us, the notation we use in this paper and some necessary material from graph theory. The GBS output measurement probability is
\begin{equation}\label{eq:ProbMixedGBSdisplace}
  p(\n)={e^{-{1\over2}\bs{D}^\dagger\s_Q^{-1}\bs{D}}\over{\n!}\sqrt{\det{\s_Q}}}\partial^{|\n|}_{\bs{\b},\ol{\bs\b}}
  e^{{1\over2}\bs{\g}^\top C\bs{\g}+\bs{D}^\dg\s_Q^{-1}\bs{\g}}\big\rvert_{\bs{\g}=0},
\end{equation}
where $\n=(n_i)\in\bbZ^M_{\geq0}$, $\partial^{|\n|}_{\bs{\b},\ol{\bs\b}}\equiv\prod_{i=1}^{M}{\partial^{n_i}\over\partial\b_i^{n_i}}{\partial^{n_i}\over\partial\ol{\b}_i^{n_i}}$, $\n!\df n_1!\times\dots\times n_M!$, $C=C^\top\in\bbR^{2M\times2M}$ and $\bs{\g}\df(\bs{\b},\ol{\bs{\b}})=(\b_1,\dots,\b_M,\ol\b_1,\dots,\ol\b_M)\in\bbC^{2M}$ which we view as a column vector and $\bs{D}\in\bbC^{2M}$ is a displacement $2M$-tuple. We denote
\begin{equation}\label{eq:X}
X_{2M}=\begin{bmatrix}
           0 & \bbI_M \\
           \bbI_M & 0 \\
         \end{bmatrix}.
\end{equation}
Then
\begin{equation}\label{eq:sigmaQ}
\s_Q=(\bbI_{2M}-X_{2M}C)^{-1},
\end{equation}
where the Gaussian covariance matrix $\s$ describing the state can be obtained by $\s=\s_Q-\bbI_{2M}/2$.

Let $C=A\oplus A$ be the so-called pure case scenario~\cite{bradler2018gaussian}. Then, if  $\bs{D}=(\bs{d},\ol{\bs{d}})$ for $\bs{d}\in\bbC^M$ Eq.~\eqref{eq:ProbMixedGBSdisplace} factorizes:
\begin{equation}\label{eq:ProbMixedGBSdisplaceForAopA}
    p(\n)={e^{-{1\over2}\bs{D}^\dagger\s_Q^{-1}\bs{D}}\over{\n!}\sqrt{\det{\s_Q}}}
    \big(\partial^{|\n|}_{\bs{\b}}e^{{1\over2}\bs{\b}^\top A\bs{\b}+(\bd^\top\bbI_M-\ol{\bd}^\top A)\bs{\b}}\big\rvert_{\bs{\b}=0}\big)^2,
\end{equation}
where $\partial^{|\n|}_{\bs{\b}}\equiv\prod_{i=1}^{M}{\partial^{n_i}\over\partial\b_i^{n_i}}$.

Let $G=(V,E)$ be a graph of $|V|$ vertices and $|E|$ edges. Let $A\in\bbR^{|V|\times|V|}$ denote its adjacency matrix. Given two graphs $G,H$, we write: (i) $G{\times}H$ for the tensor product of graphs, (ii) $G\,\square\,H$ for the Cartesian product of graphs, and, (iii) $G\uplus H$ for the disjoint union of graphs. In terms of their adjacency matrices $A_1$ and $A_2$ the operations correspond to: (i) $A_1\ox A_2$, (ii) $A_1\ox\bbI_2+\bbI_1\ox A_2$, and, (iii), $A_1\oplus A_2$.  A complete graph on $M$ vertices with self-loops will be denoted by $\ol{K}_n$ and its all-ones adjacency matrix by~$\bbJ_n$. In this paper,  we will exclusively use the prism graph construction $H=P_2(x)$ for the Cartesian product $G\,\square\,H$, where $P_2(x)$ is a single edge with a weight~$x$. Note $G\,\square\,H\simeq H\,\square\,G$ and the distributive property~\cite{imrich2008topics}
\begin{equation}\label{eq:CarteDistributive}
    (G_1\uplus G_2)\,\square\,H=G_1\,\square\,H\uplus G_2\,\square\,H.
\end{equation}
The matching polynomial of $G$~\cite{hosoya1971topological,farrell1979introduction,gutman1977acyclic,heilmann1970monomers,godsil1981theory} is
\begin{equation}\label{eq:matchPoly}
  \mu_G(x)\df\sum_{r=0}^{\lfloor M/2\rfloor}(-1)^rm(G,r)x^{M-2r},
\end{equation}
where $m(G,r)$ is the number of $r$-matchings enumerating  the number of ways to choose~$r$ disjoint edges in the graph. Let $A$ be a $2p\times2p$ complex-valued symmetric matrix.  The hafnian,  $\haf{A}$, (originally introduced in~\cite{caianiello1953quantum})
is defined as
$$
\haf{A} =  \frac{1}{p! 2^p}\sum_{\sigma \in S_{2p}} \prod_{j=1}^p a_{\sigma(2j-1),\sigma(2j)},
$$
a complex weighted sum of perfect matchings.  We write $m(G,r)=\sum_{|S|=2r}\haf{A_S}$ where the sum is over all subsets $S$ of vertices with cardinality~$2r$. Hence
\begin{equation}\label{eq:matchPolyhaf}
    \mu_G(x) =  \sum_{S \subseteq V}  (-1)^{|S|/2} \haf{A_S}\,x^{M-|S|}.
\end{equation}
Of course, only the subsets $S$ of even cardinality contribute. The expression for $m(G,r)$ in terms of the hafnians of submatrices naturally generalizes to the weighted $r$-matches of weighted graphs~\cite{cvetkovic1988recent}. To fit with $r=0$, the hafnian of an empty matrix is taken to be~$1$. We will also introduce the \emph{signless} matching polynomial
\begin{equation}\label{eq:matchPolySignLess}
    \mu^+_G(x)\df\sum_{r=0}^{\lfloor M/2\rfloor}m(G,r)x^{M-2r}.
\end{equation}
We recall several  facts about the matching polynomial~\cite{lovasz2009matching,cvetkovic1988recent,shi2016graph}.
\begin{thm*}[Godsil tree~\cite{godsil1981matchings}]\label{thm:godsiltree}
    Let $T(G)$ be the path-tree graph (the so-called Godsil tree) of $G$ rooted at $\ol{v}$. Then
    \begin{equation}\label{eq:Godsil}
      {\mu_G(z)\over\mu_{G\backslash\{v\}}(z)}={\mu_{T(G)}(z)\over\mu_{T(G)\backslash\{\ol{v}\}}(z)}.
    \end{equation}
\end{thm*}
\begin{thm*}[\cite{farrell1979introduction,gutman1977acyclic}]
  The matching polynomial of a tree is equal to its characteristic polynomial.
\end{thm*}
\begin{thm*}[Weighted edge recurrence~\cite{cvetkovic1988recent}]\label{thm:wer}
  Let $G=(V,E)$ be a weighted graph where $A=[a_{ij}]$ is its adjacency matrix. If edge $e_{ij}$ is incident to the vertices $v_i$ and $v_j$ then
  \begin{equation}\label{eq:edgeRec}
    \mu_G(z)=\mu_{G-e_{ij}}(z)-a_{ij}\mu_{G\backslash\{v_i,v_j\}}(z)
  \end{equation}
  and
  \begin{equation}\label{eq:edgeRecSignless}
    \mu^+_G(z)=\mu^+_{G-e_{ij}}(z)+a_{ij}\mu^+_{G\backslash\{v_i,v_j\}}(z).
  \end{equation}
\end{thm*}
Note that unlike~\cite{cvetkovic1988recent} the edge weights $a_{ij}$ in~\eqref{eq:edgeRec} and~\eqref{eq:edgeRecSignless} are not squared.

\section{The GBS polynomial for zero displacement}\label{sec:GBS}

Any undirected graph $G$ whose adjacency matrix is $A$ can be encoded into the GBS device~\cite{bradler2018gaussian} by constructing $C=A\oplus A$. An arbitrary photon number distribution can then be, at least in principle, calculated from~\eqref{eq:ProbMixedGBSdisplaceForAopA}. We call this setting the ``pure encoding''.

In this section we define a new graph polynomial called the GBS polynomial for the pure encoding setting and zero displacement (relaxations are discussed in Sections \ref{sec:DGBS} and \ref{sec:mixed}). The GBS polynomial is constructed so that its coefficients are the probabilities of orbits of non-collision photon click patterns. At the same time, the collision orbits can be shown to be the coefficients of the GBS polynomial of an extended graph (see Section \ref{subsec:collision}). In Section \ref{subsec:properties} and~\ref{subsec:identities} we prove a number of its useful properties.

\subsection{Definition of the GBS polynomial}

The GBS polynomial is defined as follows:
\begin{defi}
    Let $g(G,r) = \sum_{|S|=2r} \haf^{\,2}{A_S}$. Then, the GBS polynomial of $G$ is
    \begin{equation}\label{eq:GBSpoly}
        \gbs_G(x) \df \sum_{r=0}^{\lfloor M/2\rfloor} (-1)^r g(G,r)\,x^{M-2r}.
    \end{equation}
    Thus
    \begin{equation}\label{eq:GBSpolyHaf}
        \gbs_G(x) = \sum_{S \subseteq V}(-1)^{|S|/2}\haf^{\,2}{A_S}\,x^{M-|S|}.
    \end{equation}
    We also define the signless GBS polynomial as
    \begin{equation}\label{eq:GBSpolySignless}
    \gbs^+_G(x)\df\sum_{r=0}^{\lfloor M/2\rfloor}g(G,r) x^{M-2r}.
    \end{equation}
\end{defi}
In both matching and GBS polynomials, the leading coefficient is always $x^M$ (corresponding to $S = \emptyset$) and the coefficient of $x^{M-2}$ is the number of edges. Note that  the coefficients of $x^{M-2r}$ in the matching polynomial and GBS polynomial of a graph are different if and only if there is some $A_S$ with $|S|=2r$ whose hafnian is greater than one.  That is true if the graph has an even cycle of length $2r$; if the graph has no even cycles, the matching and GBS polynomials are equal.

\subsection{The relation between GBS polynomial and non-collision GBS statistics}\label{subsec:collision}

The motivation for the introduction of the GBS polynomial is that its coefficients are proportional to the probabilities of certain orbits - which in turn are an important output statistics of the GBS device for a range of applications \cite{bradler2018graph, schuld2019quantum}. Let $\n=(n_i)$ where $n_i\leq1$ (the so-called collision-free condition). The probability of orbit $O_{\n}$ (the set of all permutations of $\n$) represented by $\n$ for a graph $G$ whose adjacency matrix is  $A\in\bbR^{M\times M}$  reads:
\begin{equation}\label{eq:probOrbitColFree}
  p_G(O_{\n})\df\sum_{\n\in O_{\n}}^{|O_{\n}|}p_G(\n)
  ={1\over\sqrt{\det{\s_{Q}}}}\sum_{\n\in O_{\n}}^{|O_{\n}|}\haf^{\,2}{A_S}.
\end{equation}
There are $\lfloor M/2\rfloor+1$ collision-free orbits labeled by $|\n|=\sum_in_i$. Note that $|\n|=|S|=2r$, $0\leq|S|\leq M$ for the vertex subset $S$ and so from~\eqref{eq:GBSpoly} and~\eqref{eq:GBSpolyHaf} we conclude
\begin{equation}\label{eq:GBSCoeffColfree}
    g(G,r)=\sqrt{\det{\s_{Q}}}\,p_G(O_{\n}).
\end{equation}

\subsection{The relation between GBS polynomial and collision GBS statistics}\label{subsec:properties}

The situation when we drop the collision-free condition on $\n$ (so there exist modes where $n_i>1$) has been analyzed in~\cite{bradler2018graph} and we  first recall a few basic facts necessary for this work as well. Let $n=\max{n_i}$ and $\N=(N_1,\dots,N_{nM})$. We introduce a mapping  $\n\mapsto\N$ called ``decollision'' defined as
\begin{equation}\label{eq:decollision}
  n_i\mapsto (N_{n(i-1)+1},\dots,N_{ni})=(\underbrace{0,\dots,0}_{n-n_i},\underbrace{1,\dots,1}_{n_i}).
\end{equation}
The name comes from the fact that $N_j=0,1$ and note the choice of the increasing order $N_j\leq N_{j+\ell},\forall j, \ell$. Then, Eq.~\eqref{eq:ProbMixedGBSdisplaceForAopA} becomes
\begin{equation}\label{eq:equalityHaf}
\big(\partial^{|\n|}_{\bs{\b}}e^{{1\over2}\bs{\b}^\top A\bs{\b}}\big\rvert_{\bs{\b}=0}\big)^2
=\big(\partial_{\bs{\a}}^{|\bs N|}e^{{1\over2}\bs{\a}^\top (A\ox\bbJ_n)\bs{\a}}\big\rvert_{\bs{\a}=0}\big)^2
=\haf^{\,2}{[A\sox\bbJ_{|\n|}]},
\end{equation}
where $\bs\a=(\a_i)_{i=1}^{nM}$. The operation $\sox$  on the RHS  stands for the reduced Kronecker product and was introduced in~\cite{bradler2018graph}. It succinctly summarizes the action of $\partial^{|\N|}_{\bs{\a}}\equiv\prod_{i=1}^{nM}{\partial^{N_i}\over\partial\a_i^{N_i}}$ by ignoring the rows and columns of $A\ox\bbJ_n$ corresponding to $N_i=0$ in the partial derivative. To be more specific, given $\n$, take $n=\max{n_i}$, create $A\ox\bbJ_n$ and remove $n-n_i$ rows (columns) starting from the $(n(i-1)+1)$-th row (column) of $A\ox\bbJ_n$. The resulting matrix is $A\sox\bbJ_{|\n|}\in\bbR^{|\n|\times|\n|}$ which is real symmetric whenever $A$ is. It helps us  write down the collision orbit probability in a simple way~\cite{bradler2018graph}
\begin{equation}\label{eq:probOrbitCol}
  p_G(O_{\n})=\sum_{\n\in O_{\n}}^{|O_{\n}|}p_G(\n)
  ={1\over\sqrt{\det{\s_{Q}}}}{1\over\bs{n}!}\sum_{\N}\haf^{\,2}{[A\sox\bbJ_{|\n|}]},
\end{equation}
where the click patterns $\N$ we sum over are those corresponding to the summed collision orbits $\n$ via~\eqref{eq:decollision}. Is there a natural way of forming a GBS polynomial out of these probabilities?
Clearly, $A\sox\bbJ_{|\n|}$ and $A\ox\bbJ_n$ are related in a similar way as $A_S$ and $A$ so perhaps  by investigating the collision probabilities we are secretly studying the GBS polynomial of $G\times\ol{K}_n$. Indeed, this turns out to be the case.
\begin{prop}\label{prop:GBSpolyCol}
    Consider a click pattern $\n=(n_1,\dots,n_M)$  rewritten as
    \begin{equation}\label{eq:n}
        \n=(\underbrace{0,\dots,0}_{k_0},\underbrace{1,\dots,1}_{k_1},\dots,\underbrace{n,\dots,n}_{k_n}).
    \end{equation}
    Fix $n\geq1$ such that $n_i\leq n$. Then, for $0\leq 2r\leq nM$, we write the GBS polynomial of $G\times\ol{K}_n$ as
    \begin{equation}\label{eq:GBSpolyCol}
    \gbs_{G\times\ol{K}_n}(x)=\sum_{r=0}^{\lfloor nM/2\rfloor}(-)^rg(G\times\ol{K}_n,r)x^{nM-2r},
    \end{equation}
    where the coefficients read
    \begin{equation}\label{eq:GBSCoeffsMultiProbsInDetail}
      g(G\times\ol{K}_n,r)
      =\sqrt{\det{\s_Q}}\sum_{\substack{n_1+\dots+n_M=|\n|\\n_i\leq n_{i+\ell}\\n_i\leq n}}\n!\prod_{j=0}^{n}\binom{n}{j}^{k_j}p_G(O_{\n}).
    \end{equation}
\end{prop}
\begin{proof}
    We write
    \begin{equation}\label{eq:NorbitProb}
       g(G\times\ol{K}_n,r)
       =\sum_{\N\in O_{\N}}^{|O_{\N}|}\haf^{\,2}{[A\sox\bbJ_{|\n|}]}
       =\sum_{\N\in O_{\N}}^{|O_{\N}|}
       \big(\partial_{\bs{\a}}^{|\N|}e^{{1\over2}\bs{\a}^\top (A\ox\bbJ_n)\bs{\a}}\big\rvert_{\bs{\a}=0}\big)^2,
    \end{equation}
    where
    \begin{equation}\label{eq:orbitSizeforN}
        |O_{\N}|=\binom{nM}{2r}.
    \end{equation}
    The first equality comes from the definition of the GBS polynomial, Eq.~\eqref{eq:GBSpolyHaf} (think of $G$ as $G\times\ol{K}_n$ and so it becomes~\eqref{eq:GBSpolyCol}), and the second equality follows from~\eqref{eq:equalityHaf}. We express the squared term on the RHS with the help of~\eqref{eq:ProbMixedGBSdisplaceForAopA} (recall $\bs{D}=0$ for now)
    \begin{equation}\label{eq:ProbMixed}
    \sqrt{\det{\s_Q}}\,\n!\,p_G(\bs{n})=\big(\partial^{|\n|}_{\bs{\b}}e^{{1\over2}\bs{\b}^\top A\bs{\b}}\big\rvert_{\bs{\b}=0}\big)^2
    \end{equation}
    as
    \begin{equation}\label{eq:GBSCoeffsCol}
      g(G\times\ol{K}_n,r)=\sqrt{\det{\s_Q}}\sum_{\N\in O_{\bs{N}}}^{|O_{\N}|}\n!\,p_G(\n).
    \end{equation}
    Note, however, that there seems to be a mismatch:  we are summing over all the elements of the orbit $O_{\N}$ but the summand is a function of $\n$. Moreover, unlike~\eqref{eq:probOrbitCol}, we are summing over the whole orbit of $\N$. There is a link between $\n$ and $\N$ given by the decollision transformation, Eq.~\eqref{eq:decollision}, but the counting does not match: $|O_{\N}|$ for some $\N$ does not count the permutations of $\n$ it has been obtained from. Indeed, the relation between $\N$ and $\n$ is something to be careful about. As we have seen, a pattern $\n$ determines $\N$ uniquely, but not the other way around: many permutationally \emph{inequivalent} $\n$'s contribute to an $\N$ with a fixed number of ones and zeros. Moreover, we assembled the RHS of~\eqref{eq:decollision} in an increasing order for all $n_i\in\n$ but any order is equally valid. For these reasons we kept $\n!$ inside the sum in~\eqref{eq:GBSCoeffsCol}. We fix $M$ (by the choice of $G$), $n$ and $r$ (by what coefficient $g(G\times\ol{K}_n,r)$ we are after) such that $n\geq1$ and $0\leq2r\leq nM$ and we consider all possible click patterns $\n$ satisfying
    \begin{equation}\label{eq:nN}
    2r\equiv\sum_{j=1}^{nM}N_j\equiv|\N|=\sum_{i=1}^{M}n_i\equiv|\n|\quad\mbox{ s.t }\quad n=\max{n_i}.
    \end{equation}
    The sum in~\eqref{eq:GBSCoeffsCol} splits into the sum over all restricted integer partitions of $|\n|$ (hence counting the number of orbits) and the sum over each orbit:
    \begin{equation}\label{eq:GBSCoeffsColFinal}
      g(G\times\ol{K}_n,r)
      =\sqrt{\det{\s_Q}}\sum_{\substack{n_1+\dots+n_M=|\n|\\n_i\leq n_{i+\ell}\\n_i\leq n}}\n!\prod_{j=0}^{n}\binom{n}{j}^{k_j}
      \underbrace{\sum_{\n\in O_{\n}}^{|O_{\n}|}p_G(\n)}_{p_G(O_{\n})}.
    \end{equation}
    The combinatorial coefficient $\prod_{j=0}^{n}\binom{n}{j}^{k_j}$ originates from the aforementioned fact that each $n_i$  gets mapped to an $n$-tuple on the RHS of~\eqref{eq:decollision}. There are $\binom{n}{j}$ of such maps, independently for each $j$ ($k_j$ of them), where we used~\eqref{eq:n}. The product comes from $0\leq j\leq n$.
\end{proof}
\begin{rem}
  Using the Burnside formula
    \begin{equation}\label{eq:burnside}
    \left| O_{\n}\right| = {M \choose k_1, \ldots, k_\ell} = \frac{M!}{k_0!\; k_1! \ldots k_\ell!}
    \end{equation}
 for the second summand of Eq.~\eqref{eq:GBSCoeffsColFinal} we get
  \begin{equation}\label{eq:burnsideCount}
    \sum_{\substack{n_1+\dots+n_M=|\n|\\n_i\leq n_{i+\ell}\\n_i\leq n}}\binom{M}{k_0,k_1,\dots,k_\ell}\prod_{j=0}^{n}\binom{n}{j}^{k_j}
  \end{equation}
  and this must be equal to $|O_{\N}|$ in~\eqref{eq:orbitSizeforN}.
\end{rem}

Motivated by the proof of Proposition~\ref{prop:GBSpolyCol}, we introduce the triple $\tau=(M,n,r)$ encapsulating all degrees of freedom.
\begin{exa}
  Let $\tau$=(6,2,3). So $G$ is a graph on 6 vertices and $m(G\times\ol{K}_2,3)$ is given by summing squares of hafnians of all $(2\times3=)6$-dimensional submatrices of $A\ox\bbJ_2$. There are  $\binom{12}{6}=924$ of them according to~\eqref{eq:orbitSizeforN}. The  corresponding  (permutationally inequivalent) collision orbit representatives are
  \begin{align*}
    {\n} & = (1,1,1,1,1,1), \\
    {\n} & = (0,1,1,1,1,2), \\
    {\n} & = (0,0,1,1,2,2), \\
    {\n} & = (0,0,0,2,2,2)
  \end{align*}
  and their orbit elements all map to an $\N$ with an equal number of zeros and ones according to~\eqref{eq:decollision}. From~\eqref{eq:burnsideCount} we get
  \begin{equation}
    \sum_{{\substack{n_1+\dots+n_M=|\n|\\n_i\leq n_{i+\ell}\\n_i\leq n}}}\binom{M}{k_0,k_1,\dots,k_\ell}\prod_{j=0}^{n}\binom{n}{j}^{k_j}
    =
    1\times\binom{2}{1}^6+{6!\over4!}\times\binom{2}{1}^4+{6!\over2!2!2!}\times\binom{2}{1}^2+{6!\over3!3!}=924,
  \end{equation}
  where the binomial coefficients equal to one are omitted.
\end{exa}
\begin{exa}
  Let $\tau=(6,3,4)$ and a graph in Fig.~\ref{fig:g6}.
    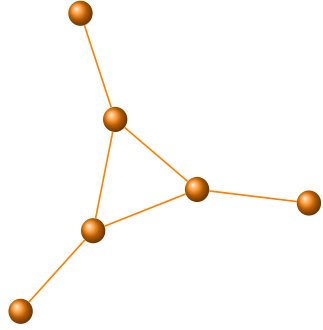
\begin{figure}[h]
        \resizebox{4.3cm}{4.3cm}{
            \begin{tikzpicture}
            \GraphInit[vstyle=Art]
            \Vertex[L=\hbox{$1$},x=1.0378cm,y=5.0cm]{v0}
            \Vertex[L=\hbox{$2$},x=1.6399cm,y=3.2192cm]{v1}
            \Vertex[L=\hbox{$3$},x=3.0601cm,y=2.0443cm]{v2}
            \Vertex[L=\hbox{$4$},x=1.2581cm,y=1.3521cm]{v3}
            \Vertex[L=\hbox{$5$},x=0.0cm,y=0.0cm]{v4}
            \Vertex[L=\hbox{$6$},x=5.0cm,y=1.8166cm]{v5}
            \Edge[](v0)(v1)
            \Edge[](v1)(v2)
            \Edge[](v1)(v3)
            \Edge[](v2)(v3)
            \Edge[](v2)(v5)
            \Edge[](v3)(v4)
            \end{tikzpicture}
        }
      \caption{A graph on six vertices.}
      \label{fig:g6}
    \end{figure}
    The $\tau$ tells us to use the following orbits
    \begin{equation}\label{eq:nExample}
        ({\n})=\big((111122),(011222),(002222),(111113),(011123),(001223),(001133),(000233)\big)
    \end{equation}
    whose proper counting should yield $\binom{nM}{2r}=\binom{18}{8}=43758$ permutations. Indeed, from Eq.~\eqref{eq:GBSCoeffsMultiProbsInDetail} we collect the combinatorial products in
    \begin{align}
      &\bs c = \\\nn
      \Bigg(
      &{6!\over2!4!}\binom{3}{2}^2\binom{3}{1}^4,{6!\over2!3!}\binom{3}{2}^3\binom{3}{1}^2,{6!\over4!2!}\binom{3}{2}^4,{6!\over5!}\binom{3}{1}^5,{6!\over3!}\binom{3}{2}\binom{3}{1}^3,{6!\over2!2!}\binom{3}{2}^2\binom{3}{1},{6!\over(2!)^3}\binom{3}{1}^2,{6!\over2!3!}\binom{3}{2}
      \Bigg)
    \end{align}
    and find $\sum_i c_i=43758$. To test~\eqref{eq:GBSCoeffsMultiProbsInDetail} itself we find the collection
    $$
    \big(\sqrt{\det{\s_Q}}\,\n!\sum_{\n\in O_{\n}}^{|O_{\n}|}p_G(\n)\big)=\big(3888,3348,4320,1296,0,96,288,60\big),
    $$
    where the coefficients $\n!$ are calculated from~\eqref{eq:nExample}. Then, the RHS of~\eqref{eq:GBSCoeffsMultiProbsInDetail} equals 384912  which is the GBS polynomial coefficient accompanying $x^{10}$ in~\eqref{eq:GBSpoly}. This is indeed equal to the RHS of~\eqref{eq:NorbitProb}.
\end{exa}

\subsection{Identities of the GBS polynomial}\label{subsec:identities}
The matching polynomial is known to satisfy a number of identities. The GBS polynomial satisfies some of them.
\begin{thm}\label{thm:GBSmultiplicative}
$\gbs_{G_1 \uplus G_2} = \gbs_{G_1} \gbs_{G_2}$.
\end{thm}
\begin{proof}
Let $V_1$ and $V_2$  be the vertex sets of $G_1$ and $G_2$, $n_1 = |V_1|$ and $n_2 = |V_2|$. For $S \subset V_1 \cup V_2$, we can write $S = S_1 \cup S_2$ where $S_1 = S \cap V_1$ and $S_2 = S \cap V_2$, and $\haf{A_S} = \haf{A_{S_1}} \haf{A_{S_2}}$ (which of course is $0$ unless $S_1$ and $S_2$ have even cardinalities).

Thus we have
\begin{align*}
    \gbs_{G_1 \uplus G_2} &= \sum_{S_1 \subseteq V_1} \sum_{S_2 \subseteq V_2} (-1)^{|S_1|/2} (-1)^{|S_2|/2} \haf^{\,2}{A_{S_1}}\haf^{\,2}{A_{S_2}} x^{n_1-|S_1|} x^{n_2 - |S_2|}\\
    &= \left( \sum_{S_1 \subseteq V_1} (-1)^{|S_1|/2}  \haf^{\,2}{A_{S_1}}  x^{n_1-|S_1|} \right)  \left( \sum_{S_2 \subseteq V_2} (-1)^{|S_2|/2}  \haf^{\,2}{A_{S_2}}  x^{n_2-|S_2|} \right)\\
    &= \gbs_{G_1} \gbs_{G_2}
\end{align*}
\end{proof}
\begin{thm}
$$
\dfrac{\dif}{\dif{x}} \gbs_G(x) = \sum_{v \in V} \gbs_{G\backslash\{v\}}(x),
$$
where $G \backslash \{v\}$ is the graph $G$ with vertex $v$ (and its incident edges) removed.
\end{thm}
\begin{proof}
\begin{align*}
    \dfrac{\dif}{\dif{x}} \gbs_G(x)
    &= \sum_{S \subseteq V} (-1)^{|S|/2} (n-|S|) \haf^{\,2}{A_S}\,x^{n-1-|S|} \\
    &= \sum_{S \subseteq V} \sum_{v \in V \backslash \{v\}} (-1)^{|S|/2} \haf^{\,2}{A_S}\,x^{n-1-|S|}\\
    &= \sum_{v \in V} \sum_{S \subseteq V \backslash \{v\}} (-1)^{|S|/2} \haf^{\,2}{A_S}\,x^{n-1-|S|}\\
    &= \sum_{v \in V} \gbs_{G\backslash \{v\}}(x)
\end{align*}
\end{proof}
The following result is not known to have a counterpart in the matching polynomial theory. It will turn out to be a special case of our main result.
\begin{thm}\label{thm:CartesianHaf}
$\gbs_G(x)=(-i)^n\haf{[G\,\square\,P_2(ix)]}$
\end{thm}
\begin{proof}
$\haf{[P_2\,\square\,G]}$ is the weighted sum of perfect matchings of $P_2\,\square\,G$.  A perfect matching of $P_2 \,\square\, G$ consists of some set of edges $(1,v), (2,v)$, say for $v \in  V \backslash S$, together with a perfect matching of $\{1\} \times S$ and a perfect matching of $\{2\} \times G$.  Of course $|S|$ must be even for this to exist.  The sum of the contributions of these to $\haf{[P_2 \,\square\, G]}$ (for a given $S \subseteq V$) is $(-ix)^{M-|S|} \haf^{\,2}{A_S}$, and multiplying  by $(-i)^M$ we get $(-1)^{|S|/2} x^{M-|S|} \haf^{\,2}{A_S}$, the term corresponding to $S$ in~\eqref{eq:GBSpolyHaf}.
\end{proof}
\begin{rem}
  Note that Theorem~\ref{thm:CartesianHaf} becomes $\gbs^+_G(x)=\haf{[G\,\square\,P_2(x)]}$ for the signless GBS polynomial \eqref{eq:GBSpolySignless}. In fact, as we will see in the next sections, the signless GBS polynomial will make frequent appearances.
\end{rem}
Theorem~\ref{thm:CartesianHaf} has an interesting consequence.
\begin{cor}\label{cor:linearSystem}
  Given $\lfloor M/2\rfloor$  hafnians $\haf{[G\,\square\,P_2(x)]}$ for known choices of $x$ we can find all the GBS coefficients $g(G,r)$ for any graph $G$ by solving a system of $\lfloor M/2\rfloor$ linear equations.
\end{cor}
As a matter of fact, the saving of the computational time can be considerable both in the exact and approximate way of  obtaining the GBS polynomial,~Eq.~\eqref{eq:GBSpolySignless}, compared to the brute force  when $2^{M-1}$ hafnians of various sizes have to be calculated. The relatively minor price to pay is the doubled size of the matrices whose hafnian we have to evaluate. The main result of this paper (Theorem~\ref{thm:duality}) leads to a significant generalization of this Corollary.

Some of the classically intractable quantities  are known to be feasible for graphs of a low treewidth (such as the permanent~\cite{cifuentes2016efficient}). The hafnian of a low treewidth graph turns out to be tractable as well. As we show in the following result, the prism over such a graph remains low-treewidth and therefore tractable too.
\begin{thm}
Suppose $G$ is a graph of treewidth $\tau$.  Then $G\,\square\,P_2$ has treewidth at most $2\tau+1$.
\end{thm}
\begin{proof}
  Since the treewidth of $G$ is $\tau$, we can construct a tree $T$ whose nodes have ``bags'', i.e., subsets  of vertices of $G$, such that
    \begin{enumerate}
    \item  Each vertex is in at least one bag.
    \item  There are at most $\tau+1$ vertices in each bag.
    \item If $(i,j)$ is an edge of $G$, there is some bag containing both $i$ and $j$.
    \item The nodes whose bags contain a given vertex form a connected subtree of $T$.
    \end{enumerate}
$G\,\square\,P_2$ has two copies $v_1$ and $v_2$ of each vertex $v$ of $G$, with edges $(v_1, w_1)$ and $(v_2, w_2)$ for each edge $(v,w)$ of $G$ and $(v_1,v_2)$ for each vertex $v$ of $G$. A tree $\tilde{T}$ for $G\,\square\,P_2$ can be constructed from $T$ with nodes corresponding to those of $T$, where the bag for each node of $\tilde{T}$ contains both copies of each vertex in the bag for the corresponding node of $T$.  Thus each bag has at most $2\tau+2$ vertices, making the treewidth $\le 2\tau+1$.
\end{proof}
We suspect that this bound is optimal. There are known lower bounds on treewidth for product graphs: see \cite{LBT}.

\begin{exa}
  Similarly to the matching polynomial case, the GBS polynomial can be given explicitly for a few prominent graph families.
  $C_n$ is the $n$-cycle graph.  For $n$ odd this has no even cycles, and the matching and GBS polynomials are the same.  For $n$ even,
  $$
  \gbs_{C_n}(x) = \mu_{C_n}(x) + 2\,(-1)^{n/2},
  $$
  as the only subgraph with hafnian greater than $1$ is $C_n$ itself. The GBS polynomial of $K_n$ can be expressed using a hypergeometric function
  $$
  \gbs_{K_n}(x) = {x}^{n}{\mbox{$_3$F$_0$}(1/2,-n/2,1/2-n/2;\,\ ;\,-4\,{x}^{-2})}.
  $$
  For the complete bipartite graph $K_{m,n}$ with parts of size $m$ and $n$ we get
  $$
  \gbs_{K_{m,n}}(x) = \sum_{r=0}^{\min{[m,n]}}(-1)^r{m \choose r} {n \choose r} (r!)^2 x^{m+n-2r}.
  $$
  In particular, in the case $m=n$, this may be expressed using a Lommel S2 function
  $$
  \gbs_{K_{n,n}}(x) = {4}^{-m}\ls2{\left[2\,m+1,0,2\,x \right]}.
  $$
  For the book graph $B_n$ on $n$ pages (consisting of $n$ $4$-cycles all with one common edge) one gets
  $$
  \gbs_{B_n}(x) = (x^2-1)^{n-1}(x^2-n-1)^2.
  $$
\end{exa}

\section{The GBS polynomial for nonzero displacement}\label{sec:DGBS}

We now turn to the case of nonzero displacement, which, curiously, uncovers relations with the original matching polynomial. This leads to a central result of this paper, namely that the displaced GBS polynomial is identical to the signless matching polynomial of an extended graph constructed from $G$ (see also Theorem \ref{thm:duality}).

\subsection{Non-collision regime}

We first consider the collision-free regime in~\eqref{eq:ProbMixedGBSdisplaceForAopA}. We offer a slightly different take on the analysis initiated in~\cite{kruse2018detailed}. Our interest in the coarse-grained probability distributions enables us to explore structures going beyond the original result. We recall that a partition $\pi$ of a set $S$ is a non-empty set of disjoint subsets $B_i$, usually called blocks, such that  their union forms $S$. For example, for $S=[5]$, $\pi(S)=((1,2),(3,4,5))=(B_1,B_2)$ is one of the partitions of the block sizes $|B_i|=2,3$. Following~\cite{kruse2018detailed}, where Eq.~\eqref{eq:ProbMixedGBSdisplace} has been analyzed with the help of Proposition~1 from~\cite{hardy2006combinatorics}, it was noticed that due to the at most quadratic argument of the exponential, the partial derivatives split into a sum over all partitions of the block size of at most two. In the case of the partial derivative from~\eqref{eq:ProbMixedGBSdisplaceForAopA} we seek to partition $[M]$ (since $\dim{A}=M$). There are $\lfloor M/2\rfloor+1$ of such partitions and we can count their size quite easily since this is just the size of the cycle conjugacy class of the symmetric group $S_{M}$ composed of cycles of the size of at most two. It is given by
\begin{equation}\label{eq:1and2cycles}
{M!\over(M-2\ell)!\ell!2^\ell},
\end{equation}
for $0\leq\ell\leq\lfloor M/2\rfloor$. Our main goal, however, is to generalize the orbit probability, Eq.~\eqref{eq:probOrbitColFree}, to include a displacement so that we can construct a \emph{displaced GBS polynomial} by properly taking into account the contributions from the blocks (cycles) of size one.

We set $\bd\in\bbR^M$ in Eq.~\eqref{eq:ProbMixedGBSdisplaceForAopA} and denote $\z^\top=\bd^\top(\bbI_M-A)$, where $A$ is the adjacency matrix of~$G$\footnote{The variable $z$ can be made complex. The consequences are yet to be explored.}. It turns out that a matching polynomial appears in the generalization of the GBS polynomial that includes a displacement. The key piece is the following lemma.
\begin{lem}\label{lem:DGBSasMPs}
  Let $A\in\bbR^{M\times M}$ be the adjacency matrix of $G$ and assume $z=z_i,\forall i$. Then
  \begin{equation}\label{eq:MatchToProb}
  \sum_{\n\in O_{\n}}^{|O_{\n}|}
  \big(\partial^{|\n|}_{\bs{\b}}e^{{1\over2}\bs{\b}^\top A\bs{\b}+\z^\top\bs{\b}}\big\rvert_{\bs{\b}=0}\big)
  =
  \sum_{\n\in O_{\n}}^{|O_{\n}|}\mu^+_{G\backslash\{V^c_{{\n}}\}}(z),
  \end{equation}
  where  $\mu^+_G(z)$ is the signless  matching polynomial of~$G$ and $V_{{\n}}^c$ is the complement of the vertex subset $V_{{\n}}$ indexed by $\n$:
  $$
  V_{{\n}}\df\big\{i:n_i=1\big\}.
  $$
\end{lem}
\begin{proof}
  We fix a collision-free $\n$ (recall $0\leq|\n|\leq M$) in $\partial^{|\n|}_{\bs{\b}}e^{{1\over2}\bs{\b}^\top A\bs{\b}+\z^\top\bs{\b}}\big\rvert_{\bs{\b}=0}$ and apply Proposition~1 from~\cite{hardy2006combinatorics}. As already observed in~\cite{kruse2018detailed}, for a given partition $\pi$, only blocks partitioning $[|\n|]$ of the length at most~two survive. According to~\eqref{eq:1and2cycles}, what remains is a sum of products of $\ell$ second derivatives  and $(|\n|-2\ell)$ first derivatives of the exponential argument. Since $\partial^2_{ij}({1\over2}\bs{\b}^\top A\bs{\b}+\z^\top\bs{\b})|_{\b=0}=a_{ij}$ and $\partial_{k}({1\over2}\bs{\b}^\top A\bs{\b}+\z^\top\bs{\b})|_{\b=0}=z_k$, where $i\neq j\neq k$, it follows that if we set $z=z_k,\forall k$ then we obtain a polynomial in $z$ of order $|\n|-2\ell$. The coefficient of $z^{|\n|-2\ell}$ is an $\ell$-match since it is a sum over all products of $\ell$ matrix elements $a_{ij}$ representing $\ell$ disjoint edges $e_{ij}$. There are
  $\lfloor|\n|/2\rfloor+1$ partitions whose size is given by \eqref{eq:1and2cycles}. We thus obtained the signless matching polynomial $\mu^+_{G\backslash\{V_{\widetilde{\n}}\}}(z)$ defined in~\eqref{eq:matchPolySignLess}.

  The last remaining step is to repeat the procedure for all permutations of $\n$ (the orbit of $\n$). This concludes the proof.
\end{proof}
\begin{rem}\label{rem:z}
  Setting $\z$ constant is a reasonable condition that can be experimentally achieved by appropriately tuning the displacement~$\bd$ for any~$A$. But can we assume $z=z_i$ for an arbitrary $A$? We rewrite $\z^\top\bs{\b}=\bd^\top(\bbI_M-A)\bs{\b}=\sum_{i=1}^{M}f_i(\bd)\b_i$, where $f_i(\bd)$ is a linear function in $\bd$. The condition $z=z_i$ demands $f_{i}(\bd)=f_{i+1}(\bd),\forall i<M$. This is $M-1$ constraints for $M$ unknowns $d_i$ and hence it is always possible to find a solution (in fact, infinitely many). From the physical perspective, if the values of $A$ are of the same order of magnitude (typically this is the case) the satisfying $d_i$'s are of the same magnitude as well and therefore all is under our control.
\end{rem}
\begin{cor}\label{cor:GBSToProb}
  \begin{equation}\label{eq:GBSToProb}
  \sum_{\n\in O_{\n}}^{|O_{\n}|}
  \big(\partial^{|\n|}_{\bs{\b}}e^{{1\over2}\bs{\b}^\top A\bs{\b}+\z^\top\bs{\b}}\big\rvert_{\bs{\b}=0}\big)^2
  =
  \sum_{\n\in O_{\n}}^{|O_{\n}|}\big(\mu^+_{G\backslash\{V^c_{{\n}}\}}(z)\big)^2.
  \end{equation}
\end{cor}
As in~\eqref{eq:probOrbitColFree} we introduce the collision-free orbit probability. Using~\eqref{eq:GBSToProb} we take Eq.~\eqref{eq:ProbMixedGBSdisplaceForAopA} and write
\begin{equation}\label{eq:probOrbitColFreeDispl}
  p_G(O_{\n})
  =\sum_{\n\in O_{\n}}^{|O_{\n}|}p(\n)
  ={e^{-{1\over2}\bs{D}^\top\s_Q^{-1}\bs{D}}\over\sqrt{\det{\s_Q}}}
  \sum_{\n\in O_{\n}}^{|O_{\n}|}\big(\mu^+_{G\backslash\{V^c_{{\n}}\}}(z)\big)^2.
\end{equation}
We interpreted the detection probabilities in a way that escaped the attention of~\cite{kruse2018detailed}. That is noteworthy but it is the duality proved later in this section (Theorem~\ref{thm:duality} and Theorem~\ref{thm:mixduality}) that makes it interesting, useful and worth defining. We first generalize~\eqref{eq:GBSpolyHaf}.
\begin{defi}\label{def:pureDGBS}
  We introduce the (signless) \emph{displaced GBS polynomial} of a graph $G$ as
    \begin{equation}\label{eq:DGBSpolyCFpure}
        \dgbs^+_G(x,z) = \sum_{|S|=0}^Mh(G,|S|;z)\,x^{M-|S|},
    \end{equation}
    where
    \begin{equation}\label{eq:hCoeffCFpure}
      h(G,|S|;z)=\sum_{\n:|\n|=|S|}\big(\mu^+_{G\backslash\{V^c_{{\n}}\}}(z)\big)^2,
    \end{equation}
    for $V^c_{{\n}}$ defined in Lemma~\eqref{lem:DGBSasMPs}. We further introduce
    \begin{equation}\label{eq:DGBSpolyColPure}
    \dgbs^+_{G\times\ol{K}_n}(x,z) = \sum_{|S|=0}^{nM}h(G\times\ol{K}_n,|S|;z)\,x^{nM-|S|},
    \end{equation}
    where $V$ is the vertex set of $G\times\ol{K}_n$, and
    \begin{equation}\label{eq:hCoeffPureCol}
      h(G\times\ol{K}_n,|S|;z)=\sum_{\N:|\N|=|S|}\big(\mu^+_{G\times\ol{K}_n\backslash\{V^c_{{\N}}\}}(z)\big)^2.
    \end{equation}
    We define $V_{{\N}}\df\big\{i:N_i=1\big\}$ for the decollisioned orbit $\N$ introduced in~\eqref{eq:decollision} and $V^c_{{\N}}$ is its complement.
\end{defi}
Eq.~\eqref{eq:DGBSpolyColPure} is, strictly speaking, a special case of~\eqref{eq:DGBSpolyCFpure} but we anticipate the special relationship the graph $G\times\ol{K}_n$ has with the collision regime as shown in Proposition~\ref{prop:GBSpolyCol} for the zero displacement case. Also, we merely chose to introduce  the signless version in order to have $\dgbs^+_G(x,0)=\gbs^+_G(x)$. If necessary the signed version can be defined as well.

\subsection{Collision regime}

It remains to show how the collision regime is linked to~\eqref{eq:hCoeffPureCol}.
\begin{prop}\label{prop:DGBSpolyColPur}
    Consider a click pattern $\n=(n_1,\dots,n_M)$ in the form of~\eqref{eq:n} and fix $n\geq1$ such that $n_i\leq n$. Then for $0\leq|S|\leq nM$ we get
    \begin{equation}\label{eq:DGBSCoeffsMultiProbsInDetail}
      h(G\times\ol{K}_n,|S|)
      =e^{{1\over2}\bs{D}^\top\s_Q^{-1}\bs{D}}\sqrt{\det{\s_Q}}\sum_{\substack{n_1+\dots+n_M=|\n|\\n_i\leq n_{i+\ell}\\n_i\leq n}}\n!\prod_{j=0}^{n}\binom{n}{j}^{k_j}
      p_G(O_{\n})
    \end{equation}
    are the coefficients of the displaced GBS polynomial of $G\times\ol{K}_n$.
\end{prop}
\begin{proof}
  The proof is a copy of the proof of Proposition~\ref{prop:GBSpolyCol} where, using the same notation, Eq.~\eqref{eq:equalityHaf} is generalized to
    \begin{equation}\label{eq:equalityPartDerv}
    \big(\partial^{|\n|}_{\bs{\b}}e^{{1\over2}\bs{\b}^\top A\bs{\b}+\z^\top\bs{\b}}\big\rvert_{\bs{\b}=0}\big)^2
    =\big(\partial_{\bs{\a}}^{|\N|}e^{{1\over2}\bs{\a}^\top (A\ox\bbJ_n)\bs{\a}+(\z^\top\ox\bbJ_{1,nM})\bs{\a}}\big\rvert_{\bs{\a}=0}\big)^2
    \end{equation}
  and used as in Eq.~\eqref{eq:NorbitProb}:
    \begin{equation}\label{eq:NorbitProbDisp}
       h(G\times\ol{K}_n,|S|;z)
       =\sum_{\bs{N}\in O_{\bs{N}}}^{|O_{\bs{N}}|}
       \big(\partial_{\bs\a}^{|\N|}e^{{1\over2}\bs{\a}^\top (A\ox\bbJ_n)\bs{\a}+(\z^\top\ox\bbJ_{1,nM})\bs{\a}}\big\rvert_{\bs{\a}=0}\big)^2.
    \end{equation}
    Using the generalization of~\eqref{eq:ProbMixed} from~\eqref{eq:ProbMixedGBSdisplaceForAopA} we get
    \begin{equation}\label{eq:ProbMixedDisp}
    e^{{1\over2}\bs{D}^\top\s_Q^{-1}\bs{D}}\sqrt{\det{\s_Q}}\,\n!\,p_G(\bs{n})
    =\big(\partial^{|\n|}_{\bs{\b}}e^{{1\over2}\bs{\b}^\top A\bs{\b}+\z^\top\bs{\b}}\big\rvert_{\bs{\b}=0}\big)^2
    \end{equation}
    and we deduce the equivalent of~\eqref{eq:GBSCoeffsCol} to be
    \begin{subequations}
    \begin{align}\label{eq:DGBSCoeffsColA}
      h(G\times\ol{K}_n,|S|;z)
      &=e^{{1\over2}\bs{D}^\top\s_Q^{-1}\bs{D}}\sqrt{\det{\s_Q}}\sum_{\N\in O_{\bs{N}}}^{|O_{\N}|}\n!\,p_G(\n)\\
      \label{eq:DGBSCoeffsColB}
      &=e^{{1\over2}\bs{D}^\top\s_Q^{-1}\bs{D}}\sqrt{\det{\s_Q}}\sum_{\substack{n_1+\dots+n_M=|\n|\\n_i\leq n_{i+\ell}\\n_i\leq n}}\n!\prod_{j=0}^{n}\binom{n}{j}^{k_j}
      \underbrace{\sum_{\n\in O_{\n}}^{|O_{\n}|}p_G(\n)}_{p_G(O_{\n})}.
    \end{align}
    \end{subequations}
    The rest goes on as in the proof of Proposition~\ref{prop:GBSpolyCol}.
\end{proof}
Using~\eqref{eq:GBSToProb} on the RHS of Eq.~\eqref{eq:equalityPartDerv} we write
\begin{equation}\label{eq:outOfNames}
\big(\partial_{\bs\a}^{|\N|}e^{{1\over2}\bs{\a}^\top(A\ox\bbJ_n)\bs{\a}+(\z^\top\ox\bbJ_{1,nM})\bs{\a}}\big\rvert_{\bs{\a}=0}\big)^2
=\big(\mu^+_{G_{A\ox\bbJ_n}\backslash\{V^c_{{\N}}\}}(z)\big)^2
\end{equation}
and it helps us obtain the equivalent of~\eqref{eq:probOrbitColFreeDispl}:
\begin{equation}\label{eq:probOrbitColDisplPure}
  p_G(O_{\n})
  =\sum_{\n\in O_{\n}}^{|O_{\n}|}p(\n)
  ={e^{-{1\over2}\bs{D}^\top\s_Q^{-1}\bs{D}}\over\sqrt{\det{\s_Q}}}
  {1\over\n!}\sum_{\N}\big(\mu^+_{G_{A\ox\bbJ_n}\backslash\{V^c_{{\N}}\}}(z)\big)^2
\end{equation}
(note that $G_{A\ox\bbJ_n}\backslash\{V^c_{{\N}}\}$ is the graph corresponding to $A\sox\bbJ_{|\n|}$ interpreted as an adjacency matrix). Just like in~\eqref{eq:probOrbitCol}, the sum over $\N$ corresponds to the collision orbits $\n$ via~\eqref{eq:decollision}.
\begin{exa}[Collision-free DGBS for $K_M$ and $\ol{K}_M$]\label{exa:colFreeKM}
  Using~\eqref{eq:GBSToProb} we can easily construct the displaced GBS polynomial for the complete graph~$K_M$. This is because of the well-known expression for $\mu_{K_M}^+$~\cite{godsilAlgComb} in terms of the coefficients of the Hermite polynomial
    \begin{equation}\label{eq:muPlusKM}
      \mu_{K_M}^+(z)=\sum_{r=0}^{\lfloor M/2\rfloor}{M!\over(M-2r)!r!2^r}z^{M-2r}.
    \end{equation}
  Due to the complete symmetry of $K_M$ we also get $\binom{M}{|S|}$ copies of $\mu^+_{K_M\backslash\{V_{\widetilde{\n}}\}}= \mu_{K_{|S|}}^+(z)$. Hence, from~\eqref{eq:hCoeffCFpure}, we get
    \begin{equation}\label{eq:hCoeffKMexample}
      h(K_M,|S|;z)=\binom{M}{|S|}\bigg(\sum_{r=0}^{\lfloor|S|/2\rfloor}{|S|!\over(|S|-2r)!r!2^r}z^{|S|-2r}\bigg)^2.
    \end{equation}
    Since $\mu_{K_M}^+(z)=\mu_{\ol{K}_M}^+(z)$ we get
    $$
    h(\ol{K}_M,|S|;z)=h(K_M,|S|;z).
    $$
    This may look surprising at first sight. The output statistics of $K_M$ and $\ol{K}_M$ are certainly different for the same displacement. But the physical difference is buried in our definition of $\z^\top=\bs{d}^\top (\bbI_M-A)$ in~\eqref{eq:ProbMixedGBSdisplaceForAopA}. It means that to reproduce the output statistics of $K_M$ using $\ol{K}_M$ (or vice versa) one just has to adjust the displacement $M$-tuple $\bs{d}$ and the squeezing parameter $c$.
\end{exa}

\subsection{Duality between GBS and matching polynomials}

We now prove the ``pure state'' version of the main result of the paper, which unveils another close link between GBS and matching polynomials.
\begin{thm}\label{thm:duality}
\begin{equation}\label{eq:duality}
  \dgbs^+_G(x,z)=\mu^+_{G\,\square\,P_2(x)}(z)
\end{equation}
\end{thm}
\begin{proof}
We write the LHS with the help of~\eqref{eq:DGBSpolyCFpure} as
$$
    \dgbs_G^+(x,z) = \sum_{S \subseteq [1,\ldots, M]} (\mu^+(A_{S,S},z))^2 x^{M-|S|}.
$$
For the RHS we write the $2M\times2M$ adjacency matrix of $G\,\square\,P_2(x)$ as
\begin{equation}
    D(x)=
    \begin{bmatrix}
    A & x\bbI_M\cr x\bbI_M & A
    \end{bmatrix}
\end{equation}
and it is advantageous to rewrite the RHS as $\mu^+(D(x),z)$. The coefficient of $z^k$ in $\mu^+(D(x), z)$ is the sum of $\haf{D_{S,S}(x)}$ for subsets $S$ of $[1,\ldots, 2M]$ with cardinality $2M - k$. Terms with $x^{j}$ will arise from matchings of $S$ containing $j$ pairs $(s, s+M)$, $s \in S$.  The remaining  elements of $S \cap [1,\ldots, M]$ must be matched to each other, as will the remaining elements of $S \cap [M+1,\ldots, 2M]$.  Thus the coefficient of $x^j z^k$ is
\begin{equation}
\sum_{S_1 \subseteq [1,\ldots,M]: |S_1| = j\ } \sum_{\ S_2 \subseteq [1,\ldots, M] \backslash S_1\ } \sum_{\ S_3 \subseteq [1,\ldots, M] \backslash S_1: |S_2|+|S_3| = 2M-2j-k} \haf{A_{S_2,S_2}}\haf{A_{S_3,S_3}}.
\end{equation}
On the other hand, the coefficient of $x^j$ in $\dgbs_G^+(x,z)$ is the sum of  $(\mu^+(A_{T,T},z))^2$ over $T \subseteq [1,\ldots,M]$ with $|T|=M-j$. Using the definition of $\mu^+$ in~\eqref{eq:matchPolySignLess} rewritten as
$$
    \mu_G^+(z) = \sum_{S \subseteq [1,\ldots,M]}\haf{A_{S,S}}\,z^{M-|S|},
$$
the coefficient of $x^j z^k$ is
\begin{equation}
\sum_{T \subseteq  [1,\ldots, M]: |T| = M-j\ } \sum_{\ T_1, T_2 \subseteq T: |T_1| + |T_2| = 2M-2j-k} \haf{A_{T_1,T_1}}\haf{A_{T_2,T_2}}
\end{equation}
(i.e., from the $T_1$ term in $\mu^+(A_{T,T},z)$ we get $z^{|T|-|T_1|} = z^{M-j-|T_1|}$ and similarly for the $T_2$ term, so we need $M-j-|T_1| + M-j-|T_2| = k$, i.e. $|T_1|+|T_2|=2M-2j-k$). Taking $T_1 = S_2$, $T_2 = S_3$, $T = [1,\ldots,M] \backslash S_1$, we see that these are the same.
\end{proof}
\begin{rem}
  Setting $z=0$ in~\eqref{eq:duality} we recover the signless version of Theorem~\ref{thm:CartesianHaf}.
\end{rem}
Identity~\eqref{eq:duality} behaves like what could be called a duality: the LHS is a polynomial in $x$ with some physical interpretation for the second indeterminate~$z$ whereas the RHS  is the polynomial in~$z$ with an auxiliary indeterminate~$x$ (see an explicit example following Theorem~\ref{thm:mixduality}  generalizing the result to the mixed case scenario). Irrespective of how we call it, the advantage of being able to calculate the displaced GBS polynomial of~$G$ by calculating the matching polynomial of the prism over~$G$ is enormous. Generalizing Corollary~\ref{cor:linearSystem}, each DGBS coefficient $h(G,|S|;z)$ is a polynomial in $z$ and the highest one is of order $2M$. Its coefficients can be calculated if we are given $2M+1$ matching polynomials of the prism over $G$ (i.e., for different values of~$x$). Even more is possible using the recurrence formulas for the matching polynomial~\cite{sagemath}, some of which are listed in Section~\ref{sec:preliminaries}. In particular, the calculation of the matching polynomial  for larger graphs  is intractable but the boundary can be pushed by the Godsil tree or the edge recurrence formula. The limits of using these results are case-dependent. The number of subgraphs whose matching polynomial needs to be calculated increases fast and the recursive use will get us only that far. For the Godsil tree, in the case of a highly connected graph the tree grows very fast. In the most extreme case of a complete graph $K_M$ the growth is factorial in $M-1$. But this is fine -- nobody expects that the calculation of the output GBS statistics (even the coarse-grained one) becomes classically tractable.

From Theorem~\ref{thm:duality}, the multiplicative property of the matching polynomial and~\eqref{eq:CarteDistributive} we conclude
\begin{cor}
  \begin{equation}\label{eq:DGBSmultiplicative}
    \dgbs^+_{G_1\uplus G_2}(x,z)=\dgbs^+_{G_1}(x,z)\dgbs^+_{G_2}(x,z).
  \end{equation}
\end{cor}

\section{Beyond pure graph encoding -- the most general case}\label{sec:mixed}

The results from the previous section hold for any graph $G$ whose adjacency matrix $A$ gets ``doubled'': $C=A\oplus A$. But sometimes this doubling procedure is not necessary and $A$ can be encoded directly. The form of the most general~case of such graph has been uncovered in~\cite{bradler2018nonnegativity} to be
\begin{equation}\label{eq:Cgen}
A=\begin{bmatrix}
         A_{11} & A_{12} \\
         A_{12}^\top & A_{11}
\end{bmatrix},
\end{equation}
where $A_{12}\succeq0$ and $A_{ij}\in\bbR^{M\times M}$ are block matrices. Note that unlike~\cite{bradler2018nonnegativity} we consider $A$ real so that it can be interpreted as an adjacency matrix. At first sight this seems like an interesting yet limited class of matrices. But we can arrive at a subclass of~\eqref{eq:Cgen} from an entirely different direction. If $A$ of $G$ is not of the form in~\eqref{eq:Cgen} the only option to encode $G$  into the GBS device is as an adjacency matrix  $C=A\oplus A$.  Then we show in~\cite{schuld2019quantum} that if the corresponding pure Gaussian state experiences the uniform photon loss the resulting mixed Gaussian state is always described by the special case of~\eqref{eq:Cgen}. It becomes
\begin{equation}\label{eq:C}
    C_{\mathrm{loss}}=\begin{bmatrix}
             A & B \\
             B & A
    \end{bmatrix},
\end{equation}
where $A,B\in\bbR^{M\times M}$ and $A=A^\top,B=B^\top$. It is therefore highly desirable to generalize Theorem~\ref{thm:duality} to deal with the adjacency matrices of this type but, in fact, we will obtain a more general result (Theorem~\ref{thm:mixduality}) that will include~\eqref{eq:Cgen} as its special case.

Let's assume $C=C_{\mathrm{loss}}$ to be from~\eqref{eq:C} and rewrite the exponent of~\eqref{eq:ProbMixedGBSdisplace} as ${1\over2}\bs{\b}^\top A\bs{\b}+\bs{z}^\top\bs{\b}$ where $\bs{z}^\top=\bd^\top(\bbI_M-B)-\ol{\bd}^\top A$. Similarly to the Remark on p~\pageref{rem:z}, we have a choice to make $\bs{z}$ such that $z=z_i,\forall i$ and we may also set $\bs{d}\in\bbR^M$. We can now harvest the fruit of our previous labor and immediately write down the mixed version of the displaced GBS polynomial and the orbit probabilities both in the collision-free and collision regime. This is the most general coarse-grained GBS statistics one can investigate.

In the collision-free case we adapt Eq.~\eqref{eq:MatchToProb} to the mixed scenario and write
\begin{equation}\label{eq:MatchMixToProb}
    \partial^{|\n|}_{\bs{\b},\ol{\bs\b}}e^{{1\over2}\bs{\g}^\top C\bs{\g}+\bs{z}^\top\bs{\b}+\bs{z}^\top\ol{\bs{\b}}}\big\rvert_{\bs{\b},\ol{\bs\b}=0}
    =
    \mu^+_{G_C\backslash\{W^c_{{\n}}\}}(z),
\end{equation}
where $G_C$ is the graph associated with $C\in\bbR^{2M\times2M}$ in~\eqref{eq:C} interpreted as an adjacency matrix and $W^c_{{\n}}$ is the complement of the vertex set defined as
$$
W_{{\n}}=\big\{i:n_i=1\big\}\cup\big\{i+M:n_i=1\big\}.
$$
The adjacency matrix of $G_C\backslash\{W^c_{{\n}}\}$ is of the form
\begin{equation}
    C_S=
    \begin{bmatrix}
        A_S & B_S \\
        B_S & A_S \\
    \end{bmatrix},
\end{equation}
where $0\leq|S|\leq M$ and $S=\{i:n_i=1\}$ is a vertex subset as before.

In the collision case we take the RHS of~\eqref{eq:equalityPartDerv} and write an equivalent of~\eqref{eq:outOfNames}
\begin{equation}
  \partial_{\bs\a,\ol{\bs\a}}^{|\bs N|}e^{{1\over2}
  (\bs{\a},\ol{\bs{\a}})^\top (C\ox\bbJ_n)(\bs{\a},\ol{\bs{\a}})+(\z^\top\ox\bbJ_{1,nM})\bs{\a}+(\z^\top\ox\bbJ_{1,nM})\ol{\bs{\a}}}\big\rvert_{\bs{\a},\ol{\bs{\a}}=0}
  =
  \mu^+_{G_{C\ox\bbJ_{n}}\backslash\{W^c_{{\N}}\}}(z),
\end{equation}
where
$$
W_{{\N}}=\big\{i:N_i=1\big\}\cup\big\{i+nM:N_i=1\big\}.
$$
The collision-free orbit $N$ is obtained via the decollision map, Eq.~\eqref{eq:decollision}, and $G_{C\ox\bbJ_{n}}\backslash\{W^c_{{\N}}\}$ denotes the graph associated with the adjacency matrix
\begin{equation}\label{eq:Csox}
    \slashed{C}
    =\begin{bmatrix}
        A\sox\bbJ_{|\n|} & B\sox\bbJ_{|\n|} \\
        B\sox\bbJ_{|\n|} & A\sox\bbJ_{|\n|} \\
    \end{bmatrix}
\end{equation}
generalizing $C_S$.

We introduce the mixed equivalent of the displaced GBS polynomial from Definition~\ref{def:pureDGBS}.
\begin{defi}\label{def:mixedDGBS}
  The signless \emph{mixed displaced GBS polynomial} of a graph $G$ is
    \begin{equation}\label{eq:DGBSpolyCFMix}
        \mdgbs^+_G(x,z) = \sum_{|S|=0}^Mq(G,|S|;z)\,x^{M-|S|},
    \end{equation}
    where
    \begin{equation}\label{eq:qCoeffCFMix}
      q(G,|S|;z)=\sum_{\n:|\n|=|S|}\mu^+_{G\backslash\{W^c_{{\n}}\}}(z).
    \end{equation}
    Similarly, we have
    \begin{equation}\label{eq:DGBSpolyColMix}
        \mdgbs^+_{G\times\ol{K}_n}(x,z)=\sum_{|S|=0}^{nM}q(G\times\ol{K}_n,|S|;z)\,x^{nM-|S|},
    \end{equation}
    where
    \begin{equation}\label{eq:qCoeffColMix}
      q(G\times\ol{K}_n,|S|;z)=\sum_{\N:|\N|=|S|}\mu^+_{G\times\ol{K}_n\backslash\{W^c_{{\N}}\}}(z).
    \end{equation}
\end{defi}
\begin{prop}\label{prop:DGBSpolyColMix}
    Consider a click pattern $\n=(n_1,\dots,n_M)$ in the form of~\eqref{eq:n} and fix $n\geq1$ such that $n_i\leq n$. Then for $0\leq|S|\leq nM$ we get
    \begin{equation}\label{eq:DGBSCoeffsMultiProbsInDetail}
      q(G\times\ol{K}_n,|S|)
      =e^{{1\over2}\bs{D}^\top\s_Q^{-1}\bs{D}}\sqrt{\det{\s_Q}}\sum_{\substack{n_1+\dots+n_M=|\n|\\n_i\leq n_{i+\ell}\\n_i\leq n}}\n!\prod_{j=0}^{n}\binom{n}{j}^{k_j}
      p_G(O_{\n})
    \end{equation}
    are the coefficients of the displaced GBS polynomial of $G\times\ol{K}_n$.
\end{prop}
The proof is nearly identical to the proof of Proposition~\ref{prop:DGBSpolyColPur} except that instead of $p(\n)$ from~\eqref{eq:ProbMixedGBSdisplaceForAopA} we use the most general expression, Eq.~\eqref{eq:ProbMixedGBSdisplace}. Henceforth the disappearance of the square of the matching polynomial in~\eqref{eq:qCoeffCFMix} and~\eqref{eq:qCoeffColMix}. In fact, Proposition~\ref{prop:DGBSpolyColPur}, Definition~\ref{def:pureDGBS} and all their consequences are a special case of Proposition~\ref{prop:DGBSpolyColMix} for $B=0$ in~\eqref{eq:C} thanks to the multiplicativity property of the matching polynomial.
For the record, we spell out the orbit probability as the generalization of~\eqref{eq:probOrbitColDisplPure}:
\begin{equation}\label{eq:probOrbitColDisplMix}
  p_G(O_{\n})
  =\sum_{\n\in O_{\n}}^{|O_{\n}|}p(\n)
  ={e^{-{1\over2}\bs{D}^\top\s_Q^{-1}\bs{D}}\over\sqrt{\det{\s_Q}}}
  {1\over\n!}\sum_{\N}\mu^+_{G_{C\ox\bbJ_n}\backslash\{W^c_{{\N}}\}}(z).
\end{equation}

Let us present our second main result of this paper. To this end, we define $C(x)$ to be the $2M \times 2M$ matrix
\begin{equation}\label{eq:Cx}
    C(x) =
    \begin{bmatrix}
    A & B + x \bbI_M \cr B^\top + x \bbI_M & A
    \end{bmatrix}.
\end{equation}
\begin{thm}\label{thm:mixduality}
Let $A$ and $B$ be symmetric $M \times M$ matrices. Then
\begin{equation}\label{eq:mixedduality}
  \mdgbs^+_G(x,z)=\mu^+_{G(x)}(z),
\end{equation}
where $G(x)$ is the graph whose adjacency matrix is $C(x)$ and $G$ corresponds to $C(0)$.
\end{thm}
\begin{proof}
  For the purpose of the proof we define
    \begin{align}
      D(S) &= ([1,\ldots, M] \backslash S) \cup ([M+1,\ldots, 2M] \backslash (S + M))\quad\mathrm{for}\quad S \subseteq [1,\ldots, M],\\
      \mdgbs^+(A,B,x,z) &= \sum_{j=0}^M q_j(A,B,z) x^j,\\
      q_j(A,B,z) &= \sum_{S \subseteq [1,\ldots, M]: |S|=j} \mu^+(C(0)_{D(S),D(S)},z),
    \end{align}
    where the last two rows are~\eqref{eq:DGBSpolyCFMix} and~\eqref{eq:qCoeffCFMix}, respectively, rewritten in the matrix language. In the same spirit, we write $\mu^+_{G(x)}(z)$ as
    $$
    \mu^+(C(x),z) = \sum_{S \subseteq [1,\ldots, 2M]} \haf{[C(x)_{S,S}]}\,z^{2M-|S|}.
    $$
    Let $c_{j,k}$ be the coefficient of $x^j z^k$ in $\mu^+(C(x),z)$. This is the coefficient of $x^j$ in the sum of $\haf{[C(x)_{S,S}]}$ for $S$ with $|S| = 2M-k$.  Of course, for this to be nonzero, $|S|$ must be even, so $k$ is even. If $Y$ and $Z$ are $m \times m$ symmetric matrices,
    $$
    \haf{[Y+Z]} = \sum_{T \subseteq [1,\ldots, m]} \haf{Y_{T,T}} \haf{Z_{T^c,T^c}},
    $$
    where $T^c = [1,\ldots, m] \backslash T$ and the hafnian of an empty matrix is taken to be $1$.  Now in our case of Eq.~\eqref{eq:Cx}, write $C(x) = x X + C(0)$ where $X$ comes from~\eqref{eq:X} and we dropped subscript $2M$. So
    $$
    \haf{[C(x)_{S,S}]} = \sum_{T \subseteq S} \haf{[x X_{T,T}]} \haf{[C(0)_{S \backslash T,S \backslash T}]}.
    $$
    The only possible $T$ that make $\haf{[x X_{T,T}]}$ nonzero are when $T = T_1 \cup (T_1 + M)$ with $T_1 \subseteq [1, \ldots, M]$,
    in which case $\haf{[x X_{T,T}]} = x^{|T_1|}$.  Write $S = T \cup S_1$ where $T$ and $S_1$ are disjoint.  Thus
    $$
    c_{j,k} = \sum_{S' \subseteq [1, \ldots, 2M]: |S'| = 2M - 2j - k} f_j(S') \haf{[C(0)_{S',S'}]},
    $$
    where $f_j(S')$ is the number of $j$-tuples $T_1$ in $[1,\dots, M]$  with $T_1 \cap S = \emptyset$ and $(T_1 + M) \cap S = \emptyset$. On the other hand, the coefficient of $z^k$ in $\mu^+(C(0)_{D(S),D(S)},z)$ is
    $$
    \sum_{T \subseteq D(S): |T|= 2M-2j-k} \haf{[C(0)_{T,T}]}.
    $$
    Thus the coefficient of $x^j z^k$ in $\mdgbs^+_G(x,z)$ is
    $$
    c'_{j,k} = \sum_{S \subseteq [1,\ldots, M]: |S|=j}\ \sum_{T \subseteq D(S): |T| = 2M-2j-k} \haf{[C(0)_{T,T}]}.
    $$
    For a given $T \subseteq [1,\ldots, 2M]$ with $|T|=2M-2j-k$, the number of $S \subseteq [1,\ldots, M]$ with $T \subseteq D(S)$ is $f_j(T)$.  Thus we have $c'_{j,k} = c_{j,k}$, and the claim is proved.
\end{proof}
\begin{exa}
  We now illustrate the mixed duality, Eq.~\eqref{eq:mixedduality}. To show the scope of the result, we will not use $C(0)$ in~\eqref{eq:Cx} of the form $C_{\mathrm{loss}}$ in~\eqref{eq:C} corresponding to the photon loss and we will not even consider a physical adjacency matrix~\eqref{eq:Cgen}. Indeed, in the proof of Theorem~\eqref{thm:mixduality}, there is no mention of $A_{12}\equiv B\succeq0$. Recall~\cite{bradler2018nonnegativity} that a non-physical adjacency matrix means a non-physical covariance matrix derived from~\eqref{eq:sigmaQ}. Hence, let's choose graph $G$  whose adjacency matrix is
  \begin{equation}
    C(0)=\left[
    \begin{array}{cccccc}
     0 & 1 & 1 & 0 & 1 & 1 \\
     1 & 0 & 0 & 0 & 1 & 1 \\
     1 & 0 & 0 & 1 & 1 & 0 \\
     0 & 0 & 1 & 0 & 1 & 1 \\
     1 & 1 & 1 & 1 & 0 & 0 \\
     1 & 1 & 0 & 1 & 0 & 0 \\
    \end{array}
    \right].
  \end{equation}
  So $M=3$ and for the LHS of~\eqref{eq:mixedduality} we find from~\eqref{eq:qCoeffCFMix}
  \begin{align*}
    q(G,0;z) &= 1, \\
    q(G,1;z) &= 1+3z^2, \\
    q(G,2;z) &= 4 + 11 z^2 + 3 z^4,\\
    q(G,3;z) &= 5+21 z^2+10 z^4+z^6.
  \end{align*}
  Then
  \begin{equation}\label{eq:mDGBSexa}
    \mdgbs^+_G(x,z)=x^3+x^2(1+3 z^2)+x(4+11 z^2+3 z^4)+5+21 z^2+10 z^4+z^6.
  \end{equation}
  For the matching polynomial on the RHS of~\eqref{eq:mixedduality} we find from~\eqref{eq:matchPolySignLess} (for $M\mapsto2M=6$):
  \begin{align*}
    m(G(x),0) &= 1, \\
    m(G(x),2) &= 10 + 3 x, \\
    m(G(x),4) &= 21 + 11 x + 3 x^2, \\
    m(G(x),6) &= 5 + 4 x + x^2 + x^3.
  \end{align*}
  We get
  \begin{equation}\label{eq:MPexa}
    \mu^+_{G(x)}(z)=z^6+z^4(10+3x)+z^2(21+11x+3x^2)+5+4x+x^2+x^3
  \end{equation}
  and the polynomials are identical as expected.
\end{exa}

\section{Deriving novel GBS output statistics, its properties and applications}\label{sec:coarsegrainDistro}

While the coarse-grained probability distribution of orbits led to new types of matching polynomials, we will now see how investigating these polynomials leads to a new coarse-graining strategy. This strategy summarizes orbits to ``meta-orbits'' and has been successfully used by us in Ref \cite{schuld2019quantum}, a success that this more technical motivation may be able to explain.

\subsection{Motivation}

The rationale behind the process of probability coarse-graining is the fact that a single measurement event becomes highly unlikely as the number of modes of the photonic circuit incarnating GBS increases. Hence it is better to cluster some events together to ``combined events'' like orbits, and investigate the collective probability. Furthermore, the number of photon click patterns grows extremely fast with the maximum number of photons to consider, and it therefore quickly becomes unrealistic to work with this distribution as a deterministic output estimated by the device.  The distribution of photon click patterns, even though believed to be classically intractable, is therefore not ``quantum feasible'' in applications with a deterministic output. However, one has to strike the right balance: too much coarse-graining will certainly boost the probability of detection of such a ``combined'' event but it is perhaps clear that it can hardly be useful in any quantum task.

\subsection{Finding the right output distribution}

Let us first recapitulate what distributions we have encountered so far. We  consider the most general case of the displaced GBS in the collision regime. For simplicity, we assume $|\n|\leq M$ which will be typically true in an experiment for a large $M$ but nothing is expected to change in general. The first probability distribution to mention is over all click patterns $p_G(\n)$. The interesting click patterns (such as $\n=(1,\dots,1)$ corresponding to the hafnian squared of the encoded graph) are, however, classically as well as quantumly intractable. The former follows from the classical complexity arguments regarding the calculations of the permanent and hafnian and the latter from the probability of measurement estimation due to the exponential number (in $M$) of possible click patterns. The detection events that are tractable (such as the probability of the vacuum) are typically uninteresting. The next natural step is to coarse-grain over all permutations of a click pattern and study all possible orbit probabilities given by the distribution $p_G(O_{\n})=\sum_{\n\in O_{\n}}^{|O_{\n}|}p_G(\n)$. The result from~\cite{bradler2018graph} on the role of the orbit probabilities as being a complete set of graph invariants for the graph isomorphism problem suggests, among other things, that the orbit probabilities should be classically intractable. But are they accessible through the GBS device? This is closely related to the number and size of all orbits and here we will clarify the link.

The question of how many orbits for a given $|\n|$ there are is equivalent to the question of how many ways an integer $|\n|$ can be partitioned. There is no closed form for the partition number but the machinery of generating functions provides an easy answer. We construct
\begin{equation}\label{eq:genFcnIntegerPartition}
  \wp(|\n|)=\prod_{k=1}^{|\n|}{1\over1-x^k},
\end{equation}
Taylor expand it around the origin and the coefficient of $x^{|\n|}$ is the desired number of partitions of $|\n|$. The number of partitions increases exponentially with $|\n|$ (and so with $M$ since we assume $|\n|<M$ for simplicity) but as we will see later in this section, the probability of some orbits is quite high making it amenable to sampling and so it can be estimated. If we wanted to coarse-grain more, however, we would find the following result
\begin{lem}\label{lem:pn}
Let
    \begin{equation}\label{eq:probCoarsegrained}
      p_G(|\n|)\df\sum_{\n~\mathrm{s.t.\ }|\n|\mathrm{\ fixed}}p_G(\n)
      ={1\over\sqrt{\det{\s_{Q}}}}
      \sum_{\genfrac{}{}{0pt}{1}{n_1+\dots+n_M=|\n|}{n_i\leq n_{i+\ell}}}{1\over\n!}\sum_{\n\in O_{\n}}^{|O_{\n}|}\haf^{\,2}{[A\sox\bbJ_{|\n|}]}.
    \end{equation}
    Then
    \begin{equation}\label{eq:probCoarsegrainedGenFcn}
      p_G(|\n|)=\det{[\bbI_{M}-c^2A^2]}^{1/2}{1\over|\n|!}{\partial^{|\n|}\over\partial w^{|\n|}}\left(\det{[\bbI_{M}-c^2w^2A^2]^{-1/2}}\right)\big\rvert_{w=0}.
    \end{equation}
\end{lem}
\begin{rem}
  Already in~\cite[Lemma~17]{bradler2018graph} it was noted that the presence of an interferometer does not affect $p_G(|\n|)$ since the total photon number is preserved. Hence $p_G(|\n|)$ can be calculated just from the array of $M$ single-mode squeezers. Therefore, this quantity  cannot serve as a graph invariant for distinguishing co-spectral non-isomorphic graphs. Here we explicitly show that this statistics is classically tractable by using~\eqref{eq:probCoarsegrainedGenFcn} and reading off the corresponding expansion coefficient. Note that the displaced version of $p_G(|\n|)$ can be analyzed using the same proof technique leading to the same conclusion.
\end{rem}
\begin{proof}
  We consider a GBS setup as described by~\eqref{eq:ProbMixedGBSdisplaceForAopA} for $\bs{d}=0$, where $A\in\bbR^{M\times M}$, and set $|\n|=2k$. All possible states produced by the squeezers that contribute to $|\n|$ are of the form $\ket{2x_1,2x_2,\dots,2x_M}$, where $\sum_{i}{x_i}=k$. Since the total probability corresponding to all $\ket{2x_1,2x_2,\dots,2x_M}$ is a preserved quantity by the interferometer, we can omit it altogether and write it in terms of $M$ single-mode squeezers as $\b\prod_{j=1}^M\tau_{x_j}^2(r_j)$ where $c\la_j=\tanh{r_j}$ and
  \begin{equation}\label{eq:SMSSprob}
    \tau_{x_j}^2(r_j)={(2x_j)!\over2^{2x_j}(x_j!)^2}c^{2x_j}\la_j^{2x_j},
  \end{equation}
  where $0<c<1/\|A\|_2$ ($\|A\|_2$ being the spectral norm of $A$), $\la_j$ are the eigenvalues of $A$ and $\b$ is the normalization constant. The probability of $|\n|=2k$ is then
    \begin{subequations}
      \begin{align}\label{eq:pn}
        p(|\n|)=p(2k)&=\b\sum_{\genfrac{}{}{0pt}{1}{x\in\bbN^M}{\sum_ix_i=k}}\prod_{j=1}^M\tau_{x_j}^2(r_j)\\
                     &=\b\Big({c\over2}\Big)^{2k}\sum_{\genfrac{}{}{0pt}{1}{x\in\bbN^M}{\sum_ix_i=k}}\prod_{j=1}^M{(2x_j)!\over(x_j!)^2}\la_j^{2x_j}.
      \end{align}
    \end{subequations}
  Its ordinary generating function reads
  \begin{equation}\label{eq:ordGenFcn}
    G(w)=\sum_{k=0}^\infty p(2k)w^{2k}=\b\prod_{j=1}^Mg_j(w),
  \end{equation}
  where
  \begin{equation}\label{eq:g}
    g_j(w)=\sum_{x_j=0}^\infty{(2x_j)!\over(x_j!)^2}\bigg({c\la_jw\over2}\bigg)^{2x_J}={1\over\sqrt{1-c^2\la_j^2w^2}}.
  \end{equation}
  So
  \begin{equation}\label{eq:G}
    G(w)=\prod_{j=1}^M{1\over\sqrt{1-c^2\la_j^2w^2}}=\det{[\bbI_M-c^2w^2A^2]^{-1/2}}.
  \end{equation}
  It remains to find the normalization constant, for example, by setting $w=1$ in~\eqref{eq:ordGenFcn} and demanding $G(1)=1$. We find
  $$
  \b=\det{[\bbI_M-c^2A^2]^{1/2}}
  $$
  and the claim follows.
\end{proof}

\subsection{Summarizing click patterns to meta-orbits}

In order to strike the right balance between the desired classical intractability and quantum (GBS) feasibility we would like to coarse-grain more than $p_G(O_{\n})$ but less than $p_G(|\n|)$. In fact, Propositions~\ref{prop:GBSpolyCol},~\ref{prop:DGBSpolyColPur} and~\ref{prop:DGBSpolyColMix} hint at such an option. Recall that we show there how multiparticle (collision) orbits coalesce into the DGBS coefficient $h(G\times\ol{K}_n,|S|;z)$  as seen in~\eqref{eq:DGBSCoeffsColB}. We set
$$
p_G(|\n|,n)=\sum_{\substack{n_1+\dots+n_M=|\n|\\n_i\leq n_{i+\ell}\\n_i\leq n}}\n!\prod_{j=0}^{n}\binom{n}{j}^{k_j}p_G(O_{\n}).
$$
The problem is, however, that $p_G(|\n|,n)$ can't be sampled ``directly''. The culprit is the combinatorial piece $\n!\prod_{j=0}^{n}\binom{n}{j}^{k_j}$. To estimate $p_G(|\n|,n)$, one has to sample all participating $p_G(O_{\n})$ ``separately'', multiply them by the combinatorial coefficients and then sum. So it is not different from sampling the less coarse-grained distribution $p_G(O_{\n})$ which,  by definition, cannot provide less information.

Motivated by $p_G(|\n|,n)$, we instead directly sample events from what we call ``meta-orbits'' $\{|\n|,\Delta_{n}\}$,
\begin{equation}\label{eq:DeltaCoarseGr}
      p_G(|\n|,\Delta_{n})\df\sum_{\n\in\Delta_n}p_G(O_{\n}),
\end{equation}
where
\begin{equation}\label{eq:DeltaDistro}
\Delta_n=\big\{\n:\sum_{i}n_i=|\n|,(\forall i)(n_i\leq n) ,(\forall\n\exists n_i\in\n)(n_i=n)\big\}.
\end{equation}
In words, meta-orbits summarize all click patterns of a total photon number equal to $|\n|$, where no detector counts more than $n$ photons. We also call this the $\Delta$ coarse-graining strategy.

Whether it is a useful quantity again  depends on the trade-off between quantum feasibility, classical intractability and the actual usefulness. There is no rigorous proof for neither of the three items at the moment -- we will only present evidence in favor of using $p_G(|\n|,\Delta_n)$. Note that $p_G(|\n|,\Delta_{n})$ is a probability distribution since $\Delta_n$ partitions the set of all $\n$'s once $|\n|$ is chosen.

\subsection{Quantum feasibility of meta-orbits}

Let us first take a look at what coarse-grained probabilities are actually accessible in an experiment as a function of the mode number $M$ and the total photon number~$|\n|$. By accessible we mean using realistic squeezing levels leading to the probabilities whose values can be estimated with a reasonable number of samples which by the repetition rate of the GBS device translates into a time estimate. Clearly, we cannot simulate any graph but only those whose GBS polynomial/coarse-grained probabilities can be derived analytically for any graph size. The simplest case is a complete graph with loops $\ol{K}_M$. We choose $|\n|$ and calculate all $p_{\ol{K}_M}(O_{\n})$. We count the number of orbits by~\eqref{eq:genFcnIntegerPartition} (for comfort we again assume $|\n|\leq M$) and their size by the Burnside formula, Eq.~\eqref{eq:burnside}. We use it together with~\eqref{eq:muPlusKM} and $\mu_{\ol{K}_M}^+(z)=\mu_{K_M}^+(z)$ to write~\eqref{eq:probOrbitColDisplPure} explicitely as
\begin{equation}\label{eq:probKMloop}
  p_{\ol{K}_M}(O_{\n})
  ={e^{-{1\over2}\bs{D}^\top\s_Q^{-1}\bs{D}}\over\sqrt{\det{\s_Q}}}\binom{M}{k_0,k_1,\dots,k_\ell}{1\over\n!}\bigg(\sum_{r=0}^{\lfloor M/2\rfloor}{M!\over(M-2r)!r!2^r}z^{M-2r}c^r\bigg)^2,
\end{equation}
where $z=(1-c)d-(M-1)cd$ for $d=d_i,\forall i$. The factor $c^r$ appears due to the necessary ``renormalization'': $A\mapsto cA$, where $c$ is bounded by the inverse of the operator norm of $A$~\cite{bradler2018graph}. Note that unlike the collision-free case in the example on page~\pageref{exa:colFreeKM}, $p_{\ol{K}_M}(O_{\n})=p_{{K}_M}(O_{\n})$ does \emph{not} hold.

We first plot $p_{\ol{K}_{40}}(O_{\n})$ in Fig.~\ref{fig:K40prob} for a zero displacement and two values of the squeezing parameter $c$ from the allowed interval $0<c<1/40$.
\begin{figure}[t]
      \resizebox{15cm}{8cm}{\includegraphics{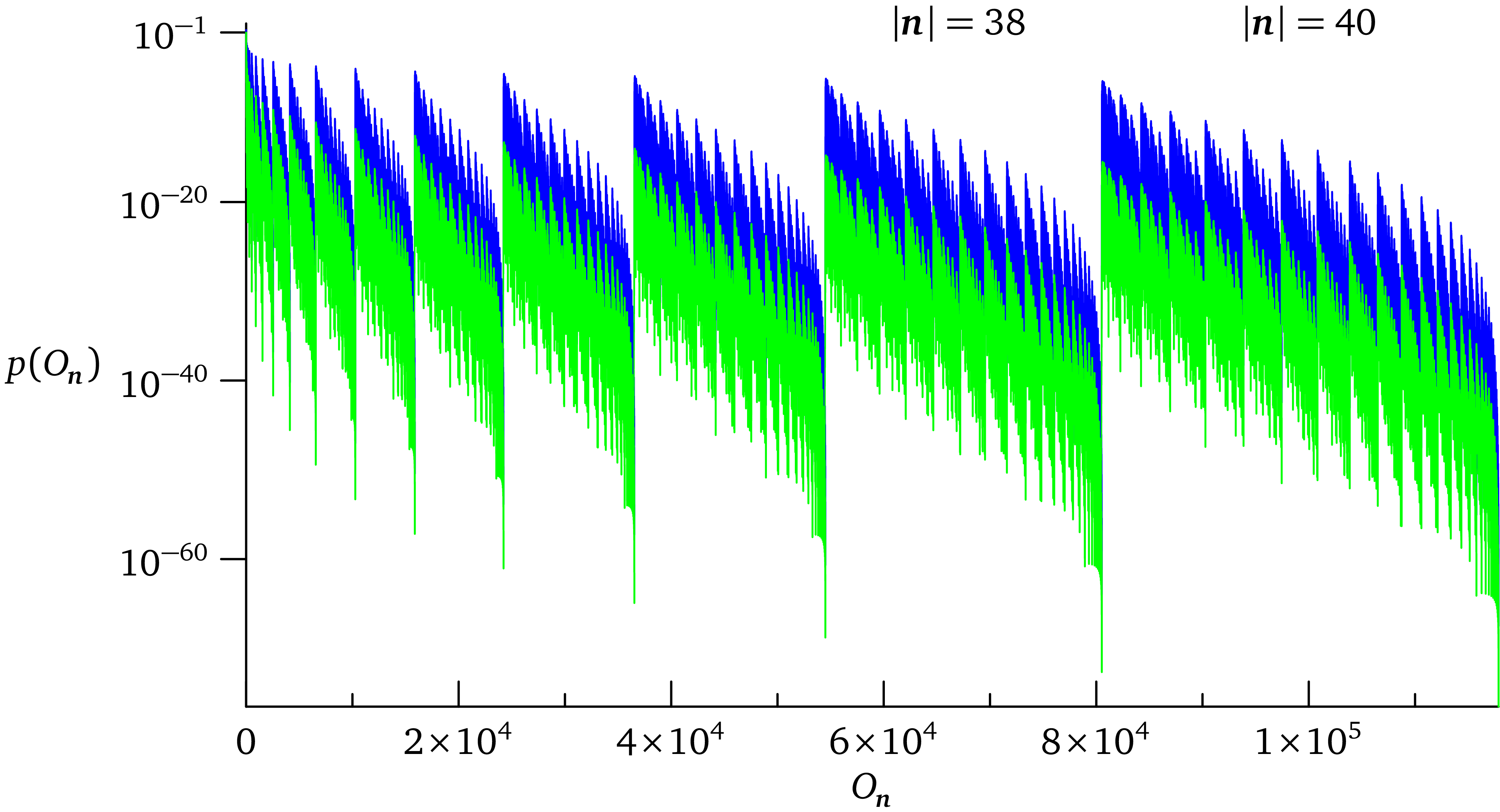}}
      \caption{Coarse-grained probability $p_{\ol{K}_M}(O_{\n})$ for $M=40$ for $|n|=0,2,\dots,40$ for $c=1/50$ (blue) and $c=1/85$ (yellow) and a zero displacement. The orbit order $O_{\n}$ on the $x$ axis is explained in the main text.}
      \label{fig:K40prob}
\end{figure}
Let us describe how the orbits on the $x$ axis are ordered. For each $|\n|$ and all $n$'s satisfying $1\leq n\leq|\n|$  we order the orbit representatives naturally in the following way: we rewrite~\eqref{eq:n} as
\begin{equation}\label{eq:nkn}
  \n=(0^{k_0},1^{k_1},\dots,n^{k_n})\equiv(1^{k_1},\dots,n^{k_n}),
\end{equation}
where $\sum_{j=1}^{|\n|}jk_j=|\n|$. The case of $n=1$ is trivial (the orbit either exists or no) so we start with $n=2$ and increase $k_2$ starting from $k_2=1$. Once all possibilities are found we set $n=3, k_3=1$ and search for all possible $k_2$'s. In this way we iteratively continue. An example of the ordering is in~\eqref{eq:nExample}.

The first thing we notice in Fig.~\ref{fig:K40prob} is the immense probability range of all orbits. Second, even the most likely orbit farther from the vacuum are quite unlikely even for a relatively high squeezing. The values $c=1/50$ and $c=1/85$ correspond to $9.5\,\mathrm{dB}$ and  $4.4\,\mathrm{dB}$ of squeezing, respectively. We address the second point in a moment.
\begin{figure}[h]
      \resizebox{15cm}{8cm}{\includegraphics{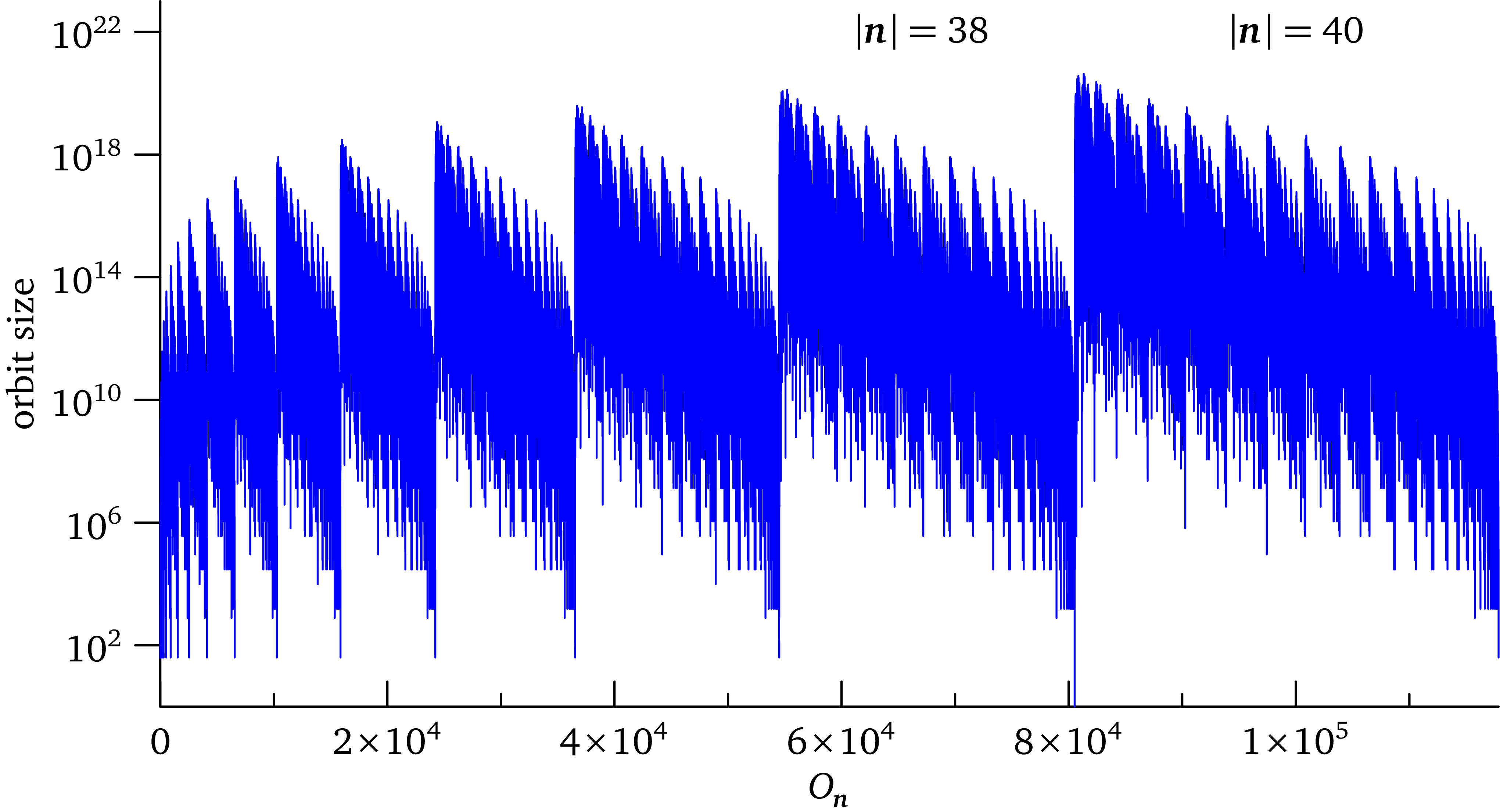}}
      \caption{The orbit sizes are clearly correlated with the orbit probabilities (cf.~Fig.~\ref{fig:K40prob} for a given $|\n|$). The orbit order on the $x$ axis is explained in the text.}
      \label{fig:K40orbits}
\end{figure}
The explanation for the large probability range lies in the orbit sizes shown in Fig.~\ref{fig:K40orbits}. We can see that for a given $|\n|$ the shape of the orbit function nearly perfectly copies the probability landscape. Indeed, this is an expected behavior. Following a reasoning from statistical mechanics, the PNR output for the click patterns with a lot of permutations is more likely (for a fixed $|\n|$). But this correspondence is not exact as can also be  seen by zooming in and comparing with the maxima for $|\n|=38$ or 40 in Figs.~\ref{fig:K40prob} and~\ref{fig:K40orbits}. Despite the fact that they do not exactly match, they are pretty close and the procedure will help us localize the orbits with the highest probabilities.

We are interested in maximizing $|O_{\n}|$ in Eq.~\eqref{eq:burnside} (or equivalently minimizing $k_1!k_2!\ldots k_\ell!$) for fixed $M$ and $|{\n}|$. We can formulate the problem as a binary integer linear programming problem (essentially a two-dimensional knapsack problem~\cite{martello1990knapsack}). Let $m$ be the largest $i$ for which we want to consider the possibility of $k_i \ne 0$.  For $0 \le i  \le m$ let $M_i$ be the largest value of $k_i$ that we want to consider. Then we take binary variables $x_{ij}$, $j = 1 \ldots M_i$, with the interpretation that $x_{ij} = 1$ if $k_i\geq j$.  Thus $k_i=\sum_{j=1}^{M_i}x_{ij}$.  Our objective will make it advantageous to have $x_{ij}\geq x_{i(j+1)}$. The cost of $x_{ij}$ is $\log{j}$, so that the total cost will be $\log{\left[\prod_i k_i!\right]}$.  The binary integer linear programming problem is
\begin{alignat*}{3}
 & \text{minimize} & \sum_{i=0}^m \sum_{j=2}^{M_i} \log{j}\,x_{ij}& \\
 & \text{subject to} \quad& \sum_{i=0}^m \sum_{j=1}^{M_i} x_{ij} &= M ,\\
                 &&\sum_{i=0}^m  \sum_{j=1}^{M_i} i\,x_{ij} &= |{\n}|,\\
                 && x_{ij}& \in \{0,1\},\quad & i & =0,\dots,m,\ \  j=1 ,\dots, M_j.
\end{alignat*}
Although the knapsack problem is NP-complete in general, these problems don't seem particularly difficult to solve. For example, for the situation in Fig.~\ref{fig:K40orbits}, where $M=40$, we explore $|\n|=40$. We find
$$
k_0=19,\,k_1=10,\,k_2=6,\,k_3=3,\,k_4=1,\,k_5=1.
$$
This is indeed an optimal solution where $\prod_{i=0}^{5}k_i!\approx1.9\times10^{27}$. How far are we from the most likely orbit? The most likely one is for
$$
k_0=14,\,k_1=15,\,k_2=9,\,k_3=3,
$$
where  $\prod_{i=0}^{3}k_i!\approx2.8\times10^{28}$.

How do we deal with the high squeezing nuisance? By introducing a displacement whose consequence can be seen in Fig.~\ref{fig:K40probDispl}. Here we plot $p_{\ol{K}_{40}}(O_{\n})$ given by~\eqref{eq:probKMloop} for the same squeezing ($c=1/55$ corresponding to $8\,\mathrm{dB}$) with and without a displacement. The overall shape of the distribution is very well preserved where, typically, only a few mismatches occur. But, crucially, the whole pattern is shifted up making some of the orbits very likely to be sampled with a reasonable amount of squeezing.
\begin{figure}[t]
      \resizebox{15.5cm}{8cm}{\includegraphics{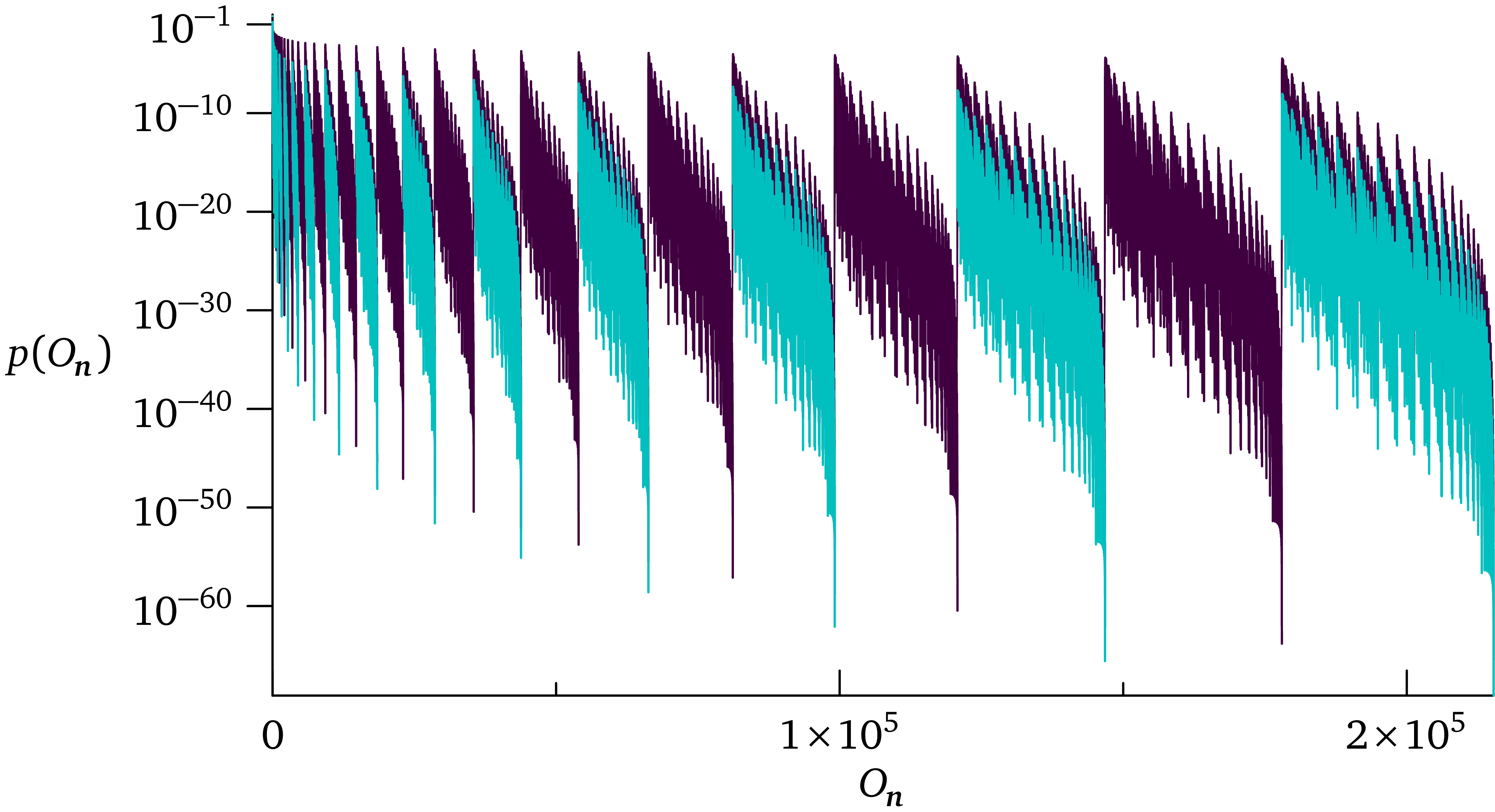}}
      \caption{We illustrate the effect of a nonzero displacement. Coarse-grained probability $p_{\ol{K}_M}(O_{\n})$ is shown for $M=40$ for $|n|=0,1,2,\dots,40$ for $c=1/55,d=0$ (cyan) and $c=1/55,d=1/2$ (magenta). The orbit order on the $x$ axis is explained in the text.}
      \label{fig:K40probDispl}
\end{figure}

One can argue that $\ol{K}_M$ is a graph too special to draw any general conclusion from it. This could be the case and so we instead plot the probability distribution corresponding to one of the two co-spectral, non-isomorphic strongly regular graph on 16 vertices~\cite{godsilAlgComb} named SRG(16,6,2,2) in Fig.~\ref{fig:SRG16} whose data were obtained from~\cite{spence2018}. The probability distribution was brute-force calculated and we will closely analyze all the consequences in the next section. For the purpose of this section we just note that the probability profile is again strongly correlated with the orbit size. The maximal squeezing necessary from top to bottom is 13.2~dB, 3.7~dB and $1\,\mathrm{dB}$.
\begin{figure}[h]
      \resizebox{15.5cm}{8.6cm}{\includegraphics[draft=false]{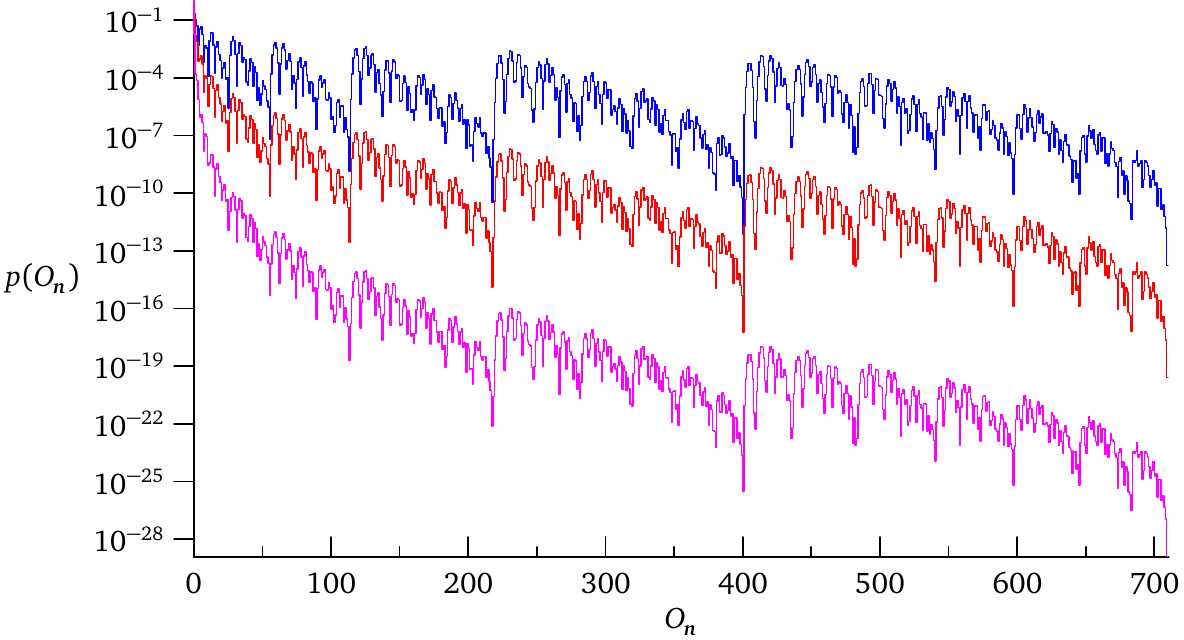}}
      \caption{Coarse-grained probability $p_{G}(O_{\n})$ where $G$ is one of the two co-spectral graphs on 16 vertices for three different values of $0<c<1/6$ and $d=0$ ($c=1/6.6, c=1/15$ and $c=1/50$, top to bottom). The orbit order on the $x$ axis is explained in the text.  Zero probability orbits were removed.}
      \label{fig:SRG16}
\end{figure}

\subsection{Classical intractability and evidence of a quantum advantage}

It is important to emphasize that all the arguments we have in favor of classical intractability are mere evidence (numerical or otherwise). The same holds for the quantum GBS advantage. The potentially good news from the previous section is that if a displacement is introduced then the orbit probabilities (whose importance has been shown on much stronger grounds~\cite{bradler2018graph}) have decent probabilities and can be used. Here, however, we will focus on the interesting properties of its more coarse-grained sibling -- the constrained probability distribution introduced in Eq.~\eqref{eq:DeltaCoarseGr}, namely $\sum_{n=1}^mp_G(|\n|,\Delta_{n})$ for a given $|\n|$ and $M$. The aim is to choose $m$ to be as high as possible because this subset typically has the highest probability as depicted in Figs.~\ref{fig:K40prob}, \ref{fig:K40probDispl} and explicitly in Fig.~\ref{fig:SRG16}. But, thanks to Lemma~\ref{lem:pn}, we know $m$ can't be equal to $\max{n_i}$ and it is not obvious if any $m$ is useful at all. Fortunately this is not the case as we will see in a moment. Let's start by asking how many orbits we are coarse-graining. This is again provided by a generating function for the restricted partitions of $|\n|$ with at most $M$ parts each of which is less or equal than $n$. It is given by the coefficient of  $x^{|\n|}$  of the Gaussian binomial coefficient~\cite{bona2012combinatorics}
\begin{equation}\label{eq:GaussBinom}
  \wp(M,n)=\binom{n+M}{M}_x=\prod_{j=1}^{M}{1-x^{n+M+1-j}\over1-x^j}
\end{equation}
expanded around $x=0$. As an example, let's reproduce the partition counting used in~\eqref{eq:nExample} for $M=6,n=3$ and $|\n|=8$. We get
$$
\wp(6,3)=1 + x + 2 x^2 + 3 x^3 + 4 x^4 + 5 x^5 + 7 x^6 + 7 x^7 + 8 x^8 + \Ocal(x^9)
$$
and the coefficient of $x^8$ is right.

Clearly, the $\Delta$ coarse-graining partitions all orbits (for a fixed $|\n|$) into a polynomial (actually linear) number of subsets in $|\n|$. Can $p_G(|\n|,\Delta_n)$ be calculated efficiently classically? We don't know but we do know that the method of Lemma~\ref{lem:pn} cannot be used. This is because no $p_G(|\n|,\Delta_{n})$ is preserved by an interferometer and so the output probability calculation must take it into account. But there seems to be a stronger argument in favor of classical intractability and this conveniently leads us to the GBS advantage topic. The numerical evidence we gathered points to the fact that $p_G(|\n|,\Delta_n)$ is helpful in a task thought to be classically intractable: the ability to distinguish co-spectral, non-isomorphic graphs, namely  strongly regular graphs, considered to belong to the hardest instances of the graph isomorphism problem. Note that GBS was used to study the graph isomorphism problem in~\cite{bradler2018graph} but there it was the orbit probability that was shown to give rise to complete graph invariants. This  can still be useful following the previous section, where for a nonzero displacement some orbit probabilities were reasonably high. But, as already mentioned, the more coarse-graining the better since such a probability distribution is easier to sample.
\begin{figure}[h]
      \resizebox{15.5cm}{8.7cm}{\includegraphics[draft=false]{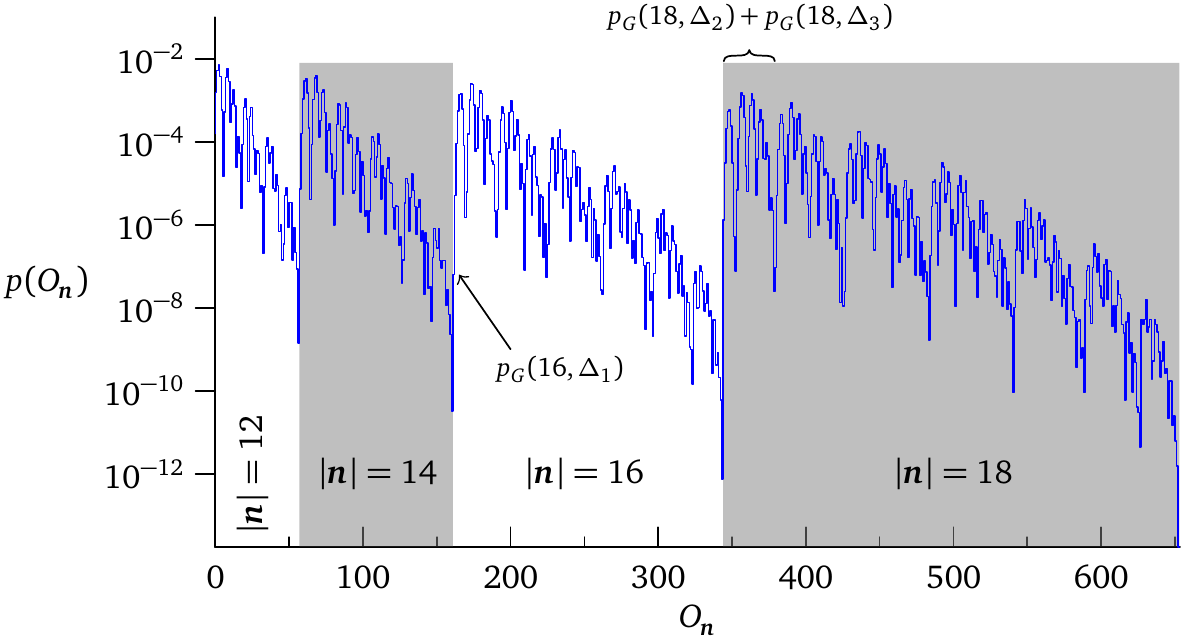}}
      \caption{Illustration of the coarse-grained probability $p_{G}(|\n|,\Delta_n)$ where $G$ is one of the two graphs from the SRG(16,6,2,2) family for $c=1/6.6$ and $d=0$.  Experimentally the most relevant $p_{G}(|\n|,\Delta_n)$'s are able to distinguish co-spectral, non-isomorphic graphs -- see the main text for details. The orbit order on the $x$ axis is explained in the text and the zero probability orbits were removed.}
      \label{fig:SRG16asy}
\end{figure}
In Fig.~\ref{fig:SRG16asy} we illustrate the performance of the $\Delta$ coarse-grained distribution introduced in~\eqref{eq:DeltaCoarseGr} on the pair of the co-spectral, non-isomorphic strongly regular graph graphs SRG(16,6,2,2). For $|\n|\leq10$ no difference is found except for $\Delta_1$. This is enough from the theoretical point of view but these orbits have in general a low probability of detection (illustrated for $|\n|=16$ -- note that $p_G(|\n|,\Delta_1)\equiv p_G(O_{\n})$). But the situation changes for $|\n|=12$ and gets better as $|\n|$ increases. For $|\n|=12,14$ it is $p_G(|\n|,\Delta_1)+p_G(|\n|,\Delta_2)$ that differs and for $|\n|=16$ it is $p_G(|\n|,\Delta_1)+p_G(|\n|,\Delta_2)+p_G(|\n|,\Delta_3)$. For $|\n|=18$ the difference is in $p_G(|\n|,\Delta_2)+p_G(|\n|,\Delta_3)$ (note that $p_G(18,\Delta_1)=0$ since $M=16$). Finally, for $|\n|=20$ we find the difference for $p_G(|\n|,\Delta_2)+p_G(|\n|,\Delta_3)+p_G(|\n|,\Delta_4)$.  As depicted for $|\n|=18$ in Fig.~\ref{fig:SRG16asy}, these are among the most likely events to detect and they carry a substantial amount of the total detection probability.

The GBS polynomial of a graph is motivated by the existence of the graph matching polynomial and its rich theory. But despite the surprising link of the displaced GBS polynomial with the matching polynomial of its prism uncovered by Theorem~\ref{thm:duality}, the displaced GBS polynomial of $G$ is different from the matching polynomial of~$G$. Is one more powerful than the other when it comes to distinguishing similar or even co-spectral graphs?  Here we show that the GBS polynomial outperforms the matching polynomial in deciding whether two co-spectral graphs are isomorphic. Let's take two co-spectral regular graphs on ten vertices~\cite{little2006combinatorial} depicted in~Fig.~\ref{fig:RG10}.
\begin{figure}[h]
\begin{tikzpicture}
\GraphInit[vstyle=Art]
\Vertex[L=\hbox{$1$},x=0.0cm,y=1.8416cm]{v0}
\Vertex[L=\hbox{$2$},x=5.0cm,y=3.5517cm]{v1}
\Vertex[L=\hbox{$3$},x=0.4538cm,y=0.8531cm]{v2}
\Vertex[L=\hbox{$4$},x=4.6416cm,y=4.695cm]{v3}
\Vertex[L=\hbox{$5$},x=2.1592cm,y=0.0cm]{v4}
\Vertex[L=\hbox{$6$},x=3.1211cm,y=5.0cm]{v5}
\Vertex[L=\hbox{$7$},x=0.7398cm,y=3.4638cm]{v6}
\Vertex[L=\hbox{$8$},x=4.7672cm,y=1.4807cm]{v7}
\Vertex[L=\hbox{$9$},x=2.2441cm,y=3.4079cm]{v8}
\Vertex[L=\hbox{$10$},x=2.9348cm,y=2.0904cm]{v9}
\Edge[](v0)(v2)
\Edge[](v0)(v4)
\Edge[](v0)(v6)
\Edge[](v0)(v8)
\Edge[](v1)(v3)
\Edge[](v1)(v5)
\Edge[](v1)(v7)
\Edge[](v1)(v8)
\Edge[](v2)(v4)
\Edge[](v2)(v6)
\Edge[](v2)(v8)
\Edge[](v3)(v5)
\Edge[](v3)(v7)
\Edge[](v3)(v8)
\Edge[](v4)(v7)
\Edge[](v4)(v9)
\Edge[](v5)(v6)
\Edge[](v5)(v9)
\Edge[](v6)(v9)
\Edge[](v7)(v9)
\end{tikzpicture}
\hspace{1cm}
\begin{tikzpicture}
\GraphInit[vstyle=Art]
\Vertex[L=\hbox{$1$},x=5.0cm,y=3.3273cm]{v0}
\Vertex[L=\hbox{$2$},x=0.4201cm,y=4.2536cm]{v1}
\Vertex[L=\hbox{$3$},x=0.0cm,y=1.9381cm]{v2}
\Vertex[L=\hbox{$4$},x=4.558cm,y=1.0396cm]{v3}
\Vertex[L=\hbox{$5$},x=1.3732cm,y=3.5033cm]{v4}
\Vertex[L=\hbox{$6$},x=3.8159cm,y=3.08cm]{v5}
\Vertex[L=\hbox{$7$},x=3.2029cm,y=1.409cm]{v6}
\Vertex[L=\hbox{$8$},x=1.3463cm,y=1.9237cm]{v7}
\Vertex[L=\hbox{$9$},x=2.0175cm,y=0.0cm]{v8}
\Vertex[L=\hbox{$10$},x=2.8607cm,y=5.0cm]{v9}
\Edge[](v0)(v3)
\Edge[](v0)(v5)
\Edge[](v0)(v6)
\Edge[](v0)(v9)
\Edge[](v1)(v2)
\Edge[](v1)(v4)
\Edge[](v1)(v7)
\Edge[](v1)(v9)
\Edge[](v2)(v4)
\Edge[](v2)(v7)
\Edge[](v2)(v8)
\Edge[](v3)(v5)
\Edge[](v3)(v6)
\Edge[](v3)(v8)
\Edge[](v4)(v6)
\Edge[](v4)(v9)
\Edge[](v5)(v7)
\Edge[](v5)(v9)
\Edge[](v6)(v8)
\Edge[](v7)(v8)
\end{tikzpicture}
\caption{Pair of regular co-spectral, non-isomorphic graphs on ten vertices.}
\label{fig:RG10}
\end{figure}
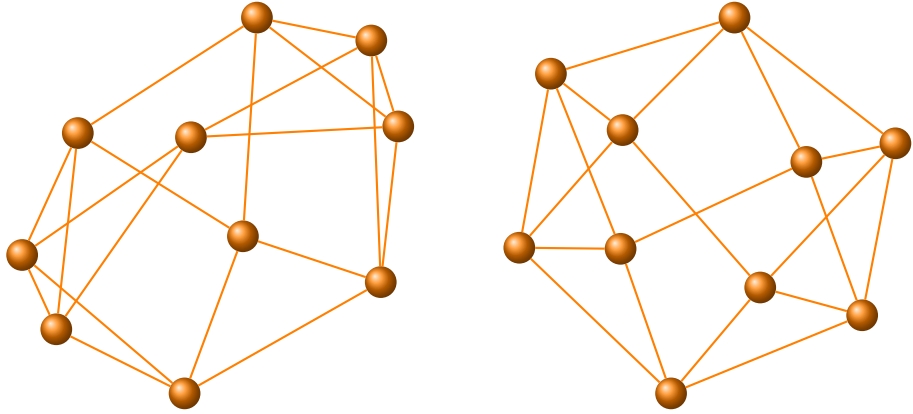
Their collision-free matching and GBS polynomials (Eqs.~\eqref{eq:GBSpolyHaf} and \eqref{eq:matchPolyhaf}) are the same, helping us conclude nothing
\begin{align}
  \mu_{G_1}(x) & = \mu_{G_2}(x) = x^{10} - 20x^8 + 130x^6 - 312x^4 + 229x^2 - 24,\\
  \gbs_{G_1}(x) & = \gbs_{G_2}(x) = x^{10} - 20x^8 + 150x^6 - 588x^4 + 1233x^2 - 576.
\end{align}
The situation in the collision regime is different already for $n=2$. We find
\begin{equation}
  \mu_{G_1\times\ol{K}_2}(x)-\mu_{G_2\times\ol{K}_2}(x) = -1536x^4 + 3840x^2 - 768
\end{equation}
for the matching polynomials showing that the graphs are not isomorphic. But the GBS polynomial, Eq.~\eqref{eq:GBSpolyCol}, performs much better:
\begin{align}
  &\gbs_{G_1\times\ol{K}_2}(x)-\gbs_{G_2\times\ol{K}_2}(x) \\
  &=2560x^{12} - 143360x^{10} + 2585600x^8 - 18898944x^6 + 40554496x^4 + 107151360x^2 - 266797056.\nonumber
\end{align}
The difference for the first time appears already for $|n|=8$ corresponding to the coefficient of~$x^{12}$. There are implications of the practical aspects of using this method. The more coarse-grained distributions differ for two non-isomorphic graphs, the more likely it is to obtain a conclusive result from sampling the graphs on a GBS device. Also, in general, it is experimentally easier to generate lower mean photon numbers of the input squeezers~\footnote{This may not be strictly true in all practical situations but it is morally correct.}. Similar results were observed for other families of co-spectral, non-isomorphic regular and strongly regular graphs.

But this is not the end of the story. If we compare the signless collision-free displaced GBS polynomial introduced in~\eqref{eq:DGBSpolyCFpure} we find a difference already there
\begin{equation}
  \dgbs^+_{G_1}(x,z)-\dgbs^+_{G_2}(x,z)=32 z^2x^3 + 16 z^2 (1 + 2 z^2)x^2 + 32z^2x.
\end{equation}
This is a witness of the power of an additional displacement having important practical consequences. In particular, if high squeezing levels to probe the large total photon numbers are difficult to achieve they can be substituted by a ``cheap'' displacement. Note that for two isomorphic graphs, their DGBS polynomials must be the same but the converse is not, in general, true.


\section{Conclusions}

This paper investigated the fruitful relation between certain coarse-grained probability distributions accessible via Gaussian boson sampling (GBS), and matching polynomials. We defined a new structure called the \emph{(displaced) GBS polynomial} of a graph encoded in the GBS device. In the collision-free regime (at most one photon per mode), its coefficients are the probabilities of a detection event and all its permutations -- so-called orbits. In the most general collision regime, the coefficients are certain natural collections of the orbits. We proved the equivalence of the displaced GBS polynomial of a graph $G$ with the matching polynomial of a different graph known as the prism over $G$. This allows us to bring the machinery of the matching polynomials -- an important topic in graph theory and theoretical physics -- to the analysis of GBS. Another consequence is a tremendous speedup of classically simulating the coarse-grained probabilities yet, at the same time, increasing our confidence in the classical intractability of the coarse-grained GBS statistics due to classical hardness results for matching polynomials. 

Using these considerations we also derive a new GBS-accessible coarse-grained probability distribution and motivate that it is classically intractable, yet useful for solving hard problems by quantum means. For this purpose, we show (but do not prove rigorously) that the coarse-grained distributions obtainable from the GBS device with the experimental parameters comfortably within today's possibilities are able to answer the graph isomorphism decision problem. We test it on several families of co-spectral, non-isomorphic (strongly) regular graphs that are considered to belong among the hardest instances.

Overall, we believe that our investigations offer a useful theoretical framework to study Gaussian boson sampling in the context of applications.

\section*{Acknowledgements}
We thank Christian Weedbrook for carefully reading the manuscript.

\bibliographystyle{unsrt}


\end{document}